\documentclass[final,leqno,showlabe]{siamltex}
\usepackage{amsmath}
\usepackage{amssymb}
\usepackage{cases}
\usepackage{enumerate}
\usepackage{graphicx}
\usepackage{cleveref}
\usepackage[usenames]{color}
\usepackage{wrapfig}
\usepackage{stmaryrd}
\usepackage{dsfont}
\SetSymbolFont{stmry}{bold}{U}{stmry}{m}{n}
\newtheorem{thm}{Theorem}[section]
\newtheorem{prop}{Proposition}[section]
\newtheorem{remark}{Remark}[section]

\newtheorem{lem}{Lemma}[section]
\newtheorem{defi}{Definition}[section]

\def\subFV{\scriptscriptstyle{FV}}
\def\subFS{\scriptscriptstyle{FS}}
\def\subVS{\scriptscriptstyle{VS}}

\renewcommand{\vec}[1]{\mathbf{#1}}

\title{Sharp-interface model for simulating solid-state dewetting in three dimensions}
\author{Wei Jiang\thanks{School of Mathematics and Statistics {\rm\&} Computational Science Hubei Key Laboratory, Wuhan University, Wuhan 430072, P.R. China (jiangwei1007@whu.edu.cn). This author's research was supported by the National Natural Science Foundation of China No. 11871384, and the Natural Science Foundation of Hubei Province Grant No. 2018CFB466.}
\and Quan Zhao\thanks{Department of Mathematics, National University of Singapore, Singapore 119076 (quanzhao90@u.nus.edu). This author's research was supported by the Ministry of Education of Singapore grant R-146-000-247-114.}
\and
Weizhu Bao\thanks{Department of Mathematics, National University of Singapore, Singapore 119076 (matbaowz@nus.edu.sg, URL: http://www.math.nus.edu.sg/\~{}bao/). This author's research was supported by the Ministry of Education of Singapore grant R-146-000-247-114 and the National Natural Science Foundation of China No. 91630207.}
}

\date{}
\begin{document}

\maketitle

\begin{abstract}
The problem of simulating solid-state dewetting of thin films in three dimensions (3D) by using a sharp-interface approach
is considered in this paper. Based on the thermodynamic variation, a speed method is used for
calculating the first variation to the total surface energy functional. The speed method shares more advantages than the traditional use of parameterized curves (or surfaces), e.g., it is more intrinsic and its variational structure
(related with Cahn-Hoffman $\boldsymbol{\xi}$-vector) is clearer. By making use of the first variation, necessary conditions for the equilibrium shape of the solid-state dewetting problem is given, and a kinetic sharp-interface model which includes the surface energy anisotropy is also proposed. This sharp-interface model describes the interface evolution in 3D which occurs through surface diffusion and contact line migration. By solving the proposed model, we perform numerical simulations to investigate the evolution of patterned films, e.g., the evolution of a cuboid and pinch-off of a long cuboid. Numerical simulations in 3D demonstrate the performance of the sharp-interface approach to capture many of the complexities observed in solid-state dewetting experiments.
\end{abstract}



\pagestyle{myheadings} \markboth{W. Jiang, Q. Zhao and W. Bao}
{Sharp-interface approach for simulating solid-state dewetting in 3D}

\section{Introduction}
\label{s1} \setcounter{equation}{0}

Driven by capillarity effects, solid thin films sitting on a substrate are rarely stable and could exhibit complex morphological changes, e.g., faceting~\cite{Jiran90,Ye10a}, edge retraction~\cite{Wong00,Dornel06,Ye11a,Kim13}, pinch-off~\cite{Jiang12,Kim15}, fingering instabilities~\cite{Kan05} and so on. This phenomenon, known as solid-state dewetting~\cite{Thompson12}, has been widely observed in many thin film\slash substrate systems~\cite{Jiran90,Kim13}. On one side, solid-state dewetting can be deleterious by fabricating the thin film structures, e.g., microelectronic and optoelectronic devices, thus destroying the reliability of the devices; on the other side, it is advantageous and can be positively used to produce the well-controlled formation of an array of micro-/nanoscale particles, e.g., used in sensors \cite{Armelao06} and as catalysts for carbon \cite{Randolph07} and semiconductor nanowire growth \cite{Schmidt09}. Recently, solid-state dewetting has attracted considerable interest, and has been widely studied by many experimental (e.g.,~\cite{Amram12,Naffouti17,Kovalenko17,Rabkin14}) and theoretical (e.g.,~\cite{Bao17b,Bao17,Dziwnik15,Jiang12,Jiang18c,Srolovitz86,Wang15,
Wong00}) research groups. The understanding of its equilibrium patterns and kinetic morphology evolution characteristics could provide important knowledge to develop new experimental methods in order to control solid-state dewetting~\cite{Leroy16}, and enhance its potential applications in thin film technologies.

Modeling solid-state dewetting has been one of active research areas and become increasingly urgent in decades. In general, surface diffusion and contact line migration have been recognized as the two main kinetic features for the evolution of solid-state dewetting~\cite{Bao17b,Jiang18}. In 1986, Srolovitz and Safran~\cite{Srolovitz86} proposed a simplified sharp-interface model to study the hole growth during the dewetting under the three assumptions, i.e., isotropic surface energy, small slope profile and cylindrical symmetry. Based on the above model, Wong {\it {et al.}} designed a ``marker particle'' numerical scheme to investigate the two-dimensional retraction of a discontinuous film (a film with a step) and the evolution of a perturbed cylindrical wire on a substrate~\cite{Wong00,Du10}. These earlier studies were focused on the isotropic surface energy, although recent experiments have demonstrated that the crystalline anisotropy could play important roles in solid-state dewetting. To include the surface energy anisotropy, many approaches have been proposed in recent years, such as a discrete model~\cite{Dornel06}, a kinetic Monte Carlo model~\cite{Pierre09b}, a crystalline model~\cite{Carter95,Zucker13} and continuum models based on partial differential equations~\cite{Bao17,Jiang16,Jiang18,Wang15}. From a mathematical perspective, theoretical solid-state dewetting studies can be categorized into two major problems: one focuses on the equilibrium of solid particles on substrates~\cite{Bao17b,Korzec14}; the other focuses on investigating the kinetic evolution of solid-state dewetting~\cite{Jiang16,Jiang18,Wang15}. In this paper, we aim at developing a sharp-interface approach for studying
these problems about solid-state dewetting in 3D.

Under isothermal conditions, the equilibrium shape for a free-standing solid particle can be formulated by minimizing the interfacial energy subject to the constraint of a constant volume:
\begin{equation}
\min_{\Omega} W:=W(S)= \int_S\gamma(\vec n)\;dS \quad\rm{s.t.}\quad|\Omega|=const,
\label{eqn:Sfenergy}
\end{equation}
where $\Omega \subset \mathbb{R}^3$ is the enclosed domain by a closed surface $S$, and $\gamma(\vec n)$ is the surface energy (density) with $\vec n=(n_1,n_2,n_3)^T$ representing the crystallographic orientation. Based on the $\gamma$-plot, the equilibrium shape can be geometrically constructed via the well-known Wulff (Gibbs-Wulff) construction~\cite{Wulff1901}. The resulted Wulff shape, is the inner convex region bounded by all planes that are perpendicular to orientation $\vec n$ and at a distance of $\gamma(\vec n)$ from the origin. The Wulff-Kaischew construction (also called as Winterbottom construction)~\cite{Kaischew50,Winterbottom67,Bao17} was subsequently proposed to handle with the case about particles on substrates by truncating the Wulff shape with a flat plane, and where
the Wulff shape is truncated depends on the wettability of the substrate. Meanwhile, many theories~\cite{Cahn74,Cahn93} demonstrated that the derivative of $\gamma(\vec n)$ plays an important role in investigating equilibrium and kinetic problems for solid particles with anisotropic surface energies. In 1972, Cahn and Hoffman developed the theory of $\boldsymbol{\xi}$-vector~\cite{Hoffman72,Cahn74} to describe the surface energy anisotropy of solid materials. It is defined based on a homogeneous extension of $\gamma(\vec n)$, i.e.,
\begin{equation}
\boldsymbol{\xi}(\vec n)=\nabla\hat{\gamma}(\vec p)\Big|_{\vec p =\vec n},\;\rm{with}\; \hat{\gamma}(\vec p)=|\vec p|\gamma\Bigl(\frac{\vec p}{|\vec p|}\Bigr),\quad\forall\vec p\in\mathbb{R}^3\backslash\{\vec 0\},
\label{eqn:xivector}
\end{equation}
where $|\vec p|:=\sqrt{p_1^2+p_2^2+p_3^2}$ for $\vec p=(p_1,p_2,p_3)^T\in\mathbb{R}^3$.
Under this extension, $\hat\gamma(\vec p)$ satisfies
\begin{equation}
\hat\gamma(\lambda\vec p) = |\lambda|\hat\gamma(\vec p),\qquad \nabla\hat{\gamma}(\vec p)\cdot\vec p = \hat{\gamma}(\vec p),\quad\forall\lambda\neq 0,\vec p\in\mathbb{R}^3\backslash\{\vec 0\}.
\label{eqn:xiequation}
\end{equation}
\begin{figure}
\centering
    \includegraphics[width=0.95\textwidth]{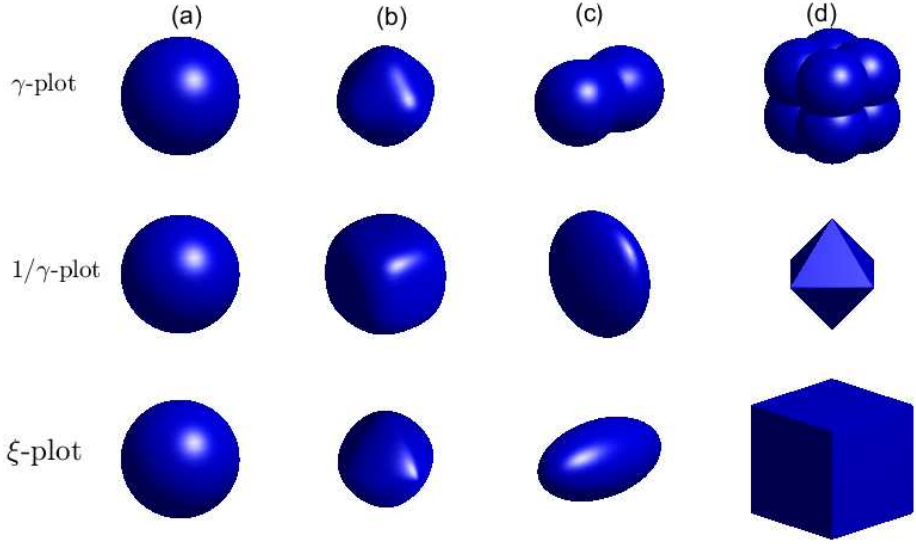}
  \caption{$\gamma$-plot, $1/\gamma$-plot and $\boldsymbol{\xi}$-plot for different surface energy anisotropies: (a) isotropic surface energy; (b) cubic anisotropic surface energy defined as $\gamma(\vec n)=1+0.3(n_1^4+n_2^4+n_3^4)$; (c) ellipsoidal surface energy $\gamma(\vec n)=\sqrt{4n_1^2+n_2^2+n_3^2}$; (d) ``cusped'' surface energy defined as $\gamma(\vec n)=|n_1|+|n_2|+|n_3|$.}
  \label{fig:anisotropy}
\end{figure}

Compared to the traditional use of scalar function $\gamma$ (or $\gamma$-plot), ~$\boldsymbol{\xi}$-vector formulation has some advantages in the description of equilibrium shapes and thermodynamic evolution for crystalline interfaces~\cite{Jiang18,Wheeler96}. From~\eqref{eqn:xiequation}, we have $\boldsymbol{\xi}\cdot\vec n=\gamma(\vec n)$,
and the magnitude of the normal component for $\boldsymbol{\xi}$ equals to $\gamma(\vec n)$. Meanwhile, $\boldsymbol{\xi}$-plot shares similar geometry with the Wulff shape, and it can be regarded as a mathematical representation of the equilibrium shape \cite{Cahn74,Peng98,Sekerka05} when its $1/\gamma$-plot is convex (i.e., weakly anisotropic). Fig.~\ref{fig:anisotropy} depicts the $\gamma$-plot, $1/\gamma$-plot and $\boldsymbol{\xi}$-plot for four different types of surface energy anisotropies: (a) isotropic surface energy, i.e., $\gamma(\vec n)\equiv 1$; (b) cubic surface energy $\gamma(\vec n) = 1 + a(n_1^4 + n_2^4 + n_3^4)$ with $a$ representing the degree of anisotropy; (c) ellipsoidal surface energy $\gamma(\vec n) = \sqrt{a_1^2 n_1^2 + a_2 ^2 n_2^2 + a_3 ^2 n_3^2}$; (d) ``cusped'' surface energy defined as $\gamma(\vec n)=|n_1|+|n_2|+|n_3|$. In the application of materials science, the surface energy could be piecewise smooth and have some ``cusped'' points, where it is not differentiable \cite{Bao17,Peng98}. A typical example is the ``cusped'' surface energy defined above. For these cases, we can regularize the surface energy with a small parameter $0<\varepsilon\ll1$ to ensure the usage of sharp-interface approach proposed in this paper, i.e.,
\begin{equation}
\gamma(\vec n) = \sqrt{\varepsilon^2+(1-\varepsilon^2)n_1^2} + \sqrt{\varepsilon^2 + (1-\varepsilon^2)n_2^2} + \sqrt{\varepsilon^2 + (1-\varepsilon^2)n_3^2}.
\label{eqn:smoothcusp}
\end{equation}
Note that $\varepsilon$ is used to smooth the surface energy, and it could relate with the width (or scale) of the rounded corner in nanoparticles~\cite{Alpay15,Bao17}.

The Cahn-Hoffman $\boldsymbol{\xi}$-vector has been recently utilized to describe the solid-state dewetting problem in two dimensions (2D)~\cite{Jiang18}. Based on the thermodynamic variation, the authors derived a sharp-interface approach via the $\boldsymbol{\xi}$-vector formulation for describing the kinetic evolution of solid-state dewetting in 2D. In this approach, the moving interface is described as a parametrization over a time-independent domain, and the variation is performed by considering an infinitesimal perturbation with respect to an open interface curve coupled with contact points~\cite{Jiang18}. However, when we want to generalize this approach to 3D, we realize that it would be very different for
calculating the thermodynamic variation for the 3D problem by using the approach of parameterized surfaces. First, the calculations of the variation in 3D via surface parametrization approach would become complicated, extremely tedious and a nightmare, and it unavoidably involves in a lot of knowledge about differential geometry; Second, for the solid-state dewetting problem, the infinitesimal perturbation to a surface in the tangential direction plays an important role in investigating the contact line migration along the substrate~\cite{Jiang16,Bao17b}, and it would make the calculations become more complicated;
Third, complicated calculations often make people easily forget the nature of the problem, and we need to investigate and make use of the variational structure of the problem. These difficulties motivate us to look for a new approach to calculating the
thermodynamic variation of solid-state dewetting in 3D. In the literature, the shape optimization problem is popular in the design of industrial structures. The speed method and shape derivatives have been widely utilized to perform the shape sensitivity analysis of shape optimization problems~\cite{Soko92,Hintermuller2004,Dougan12}. This approach avoids the parametrization of a surface and is able to deal with perturbations along arbitrary directions, and it is the desired tool we are searching for.

Therefore, based on the $\boldsymbol{\xi}$-vector formulation and the speed method, the objectives of this paper are as follows: (i) to calculate the thermodynamic variation of the energy functional for solid-state dewetting in 3D; (ii) to provide a rigorous derivation of the thermodynamic description of the equilibrium shape for solid-state dewetting in 3D; (iii) to develop a sharp-interface model which includes surface diffusion and contact line migration for simulating kinetic evolution of solid-state dewetting in 3D; and (iv) to present numerical simulations to investigate important characteristics of the morphological evolution for solid-state dewetting observed in experiments.

The rest of the paper is organized as follows. In Section 2, we briefly introduce the speed method and sharp derivatives, and then apply them for calculating the first variation of the total free energy functional. In Section 3, we rigorously derive the necessary conditions for the equilibrium shape and explicitly give an expression for the equilibrium shape by using a parametric formula. In Section 4, based on thermodynamic variation, a sharp-interface model is proposed for simulating solid-state dewetting of thin films in 3D. Subsequently, we perform some numerical simulations to demonstrate the
performance of our proposed model in Section 5. Finally, we draw some conclusions in Section 6.

\section{Thermodynamic variation}

The solid-state dewetting problem can be illustrated as Fig.~\ref{fig:dewetting3d}, where a solid thin film (in blue) can dewet  or agglomerate on a flat rigid substrate (in gray) due to capillarity effects. The total interfacial free energy of the system can be written as~\cite{Bao17b,Jiang18}
\begin{equation*}
W=\int_{S_{_{\subFV}}}\gamma_{_{\subFV}}\;dS_{_{\subFV}} + \underbrace{\int_{S_{_{\subFS}}}\gamma_{_{\subFS}}\;dS_{_{\subFS}} + \int_{S_{_{\subVS}}}\gamma_{_{\subVS}}\;dS_{_{\subVS}}}_{\rm{Substrate}\;\rm{energy}},
\end{equation*}
\begin{figure}[!htp]
\centering
\includegraphics[width=.75\textwidth]{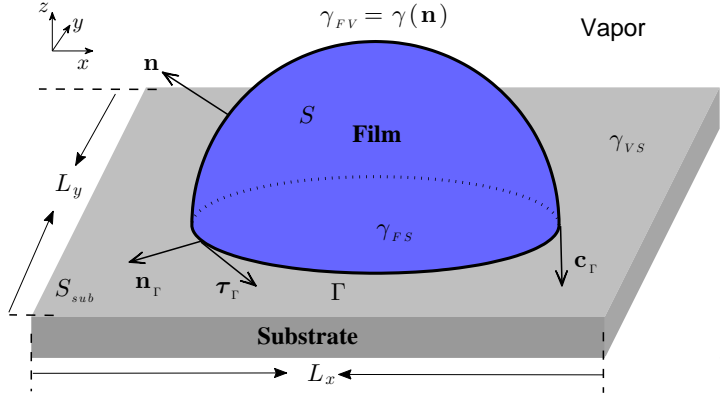}
\caption{A schematic illustration of the solid-state dewetting of a solid thin film (in blue) on a flat, rigid substrate (in gray) in 3D.}
\label{fig:dewetting3d}
\end{figure}
where $S_{_{\subFV}}:=S$, $S_{_{\subFS}}$ and $S_{_{\subVS}}$ represent the film/vapor, film/substrate and and vapor/substrate interfaces, respectively, and $\gamma_{_{\subFV}}$, $\gamma_{_{\subFS}}$ and $\gamma_{_{\subVS}}$ represent the corresponding surface energy densities. In solid-state dewetting problems, we often assume that $\gamma_{_{\subFS}},\gamma_{_{\subVS}}$ are two constants, and $\gamma_{_{\subFV}}$ is a function of the orientation of the film/vapor interface, i.e., $\gamma_{_{\subFV}}:=\gamma(\vec n)$ with $\vec n$ representing the unit normal vector of the film/vapor interface, which points outwards to the vapor phase. The film/vapor interface is here described by an open two-dimensional surface $S$ with boundary $\Gamma$ (i.e., the contact line), which is a closed plane curve on the flat substrate $S_{\rm{sub}}$.

Assume that we consider a bounded domain with size $L_x\times L_y$ on the substrate (shown in Fig.~\ref{fig:dewetting3d}). If we label the surface area enclosed by the contact line $\Gamma$ as $A(\Gamma)$, then the total interfacial free energy of the system can be calculated as
\begin{eqnarray}
W&=&\int_{S_{_{\subFV}}}\gamma_{_{\subFV}}\;dS_{_{\subFV}} + \int_{S_{_{\subFS}}}\gamma_{_{\subFS}}\;dS_{_{\subFS}} + \int_{S_{_{\subVS}}}\gamma_{_{\subVS}}\;dS_{_{\subVS}}\nonumber\\
&=&\int_S\gamma(\vec n)\;dS+(\gamma_{_{\subFS}}-\gamma_{_{\subVS}})A(\Gamma) + L_x\,L_y\,\gamma_{_{\subVS}}. \nonumber
\end{eqnarray}
By dropping off the constant term $L_xL_y\gamma_{_{\subVS}}$, we can simplify the total interfacial free energy (still labeled as $W$) as the following two parts, i.e., the film/vapor interface energy term $W_{\rm{int}}$ and the substrate energy term $W_{\rm{sub}}$,
\begin{equation}
W=W_{\rm{int}} + W_{\rm{sub}}=\int_{S}\gamma(\vec n)\;dS + (\gamma_{_{\subFS}} - \gamma_{_{\subVS}})A(\Gamma).
\label{eqn:totalenergy}
\end{equation}

As shown in Fig.~\ref{fig:dewetting3d}, we introduce three unit vectors $\vec n_{_\Gamma}$, $\boldsymbol{\tau}_{_\Gamma}$ and $\vec c_{_\Gamma}$, which are defined along the boundary $\Gamma$. More precisely, $\vec n_{_\Gamma}$ is the outer unit normal vector of the plane curve $\Gamma$ on the substrate $S_{\rm{sub}}$; $\boldsymbol{\tau}_{_\Gamma}$ is the unit tangent vector of $\Gamma$ on the substrate $S_{\rm{sub}}$, which points anticlockwise when looking from top to bottom; $\vec c_{_\Gamma}$ is called as the co-normal vector, which is normal to $\Gamma$ and tangent to the surface $S$, and points downwards. For any point $\vec x\in S$ (with $\vec x=(x_1,x_2,x_3)^T$ or $(x,y,z)^T$), if we label $\mathcal{T}_{\vec x}\,S$ as the tangent vector space to $S$ at $\vec x$, then the following properties are valid:
\begin{eqnarray}
&&\boldsymbol{\tau}_{_\Gamma}(\vec x)\in\mathcal{T}_{\vec x}\,S,\qquad\boldsymbol{\tau}_{_\Gamma}(\vec x)\sslash S_{\rm{sub}},\qquad\boldsymbol{\tau}_{_\Gamma}(\vec x)\perp\vec c_{_\Gamma}(\vec x),\quad\forall\vec x\in \Gamma,\nonumber\\
&&\vec c_{_\Gamma}(\vec x)\in\mathcal{T}_{\vec x}\,S,\qquad\vec n_{_\Gamma}(\vec x)\sslash S_{\rm{sub}},\qquad\vec n_{_\Gamma}(\vec x)\perp\boldsymbol{\tau}_{_\Gamma}(\vec x),\quad\forall\vec x\in \Gamma.   \nonumber
\end{eqnarray}

\subsection{Differential operators on manifolds}

To obtain the first variation of the above shape functional~\eqref{eqn:totalenergy}, we start by introducing some basic knowledge about surface calculus. For more details, readers could refer to~\cite{Dziuk13,Deckelnick05}.

\begin{defi}
\label{defi:calculus}
Suppose that $S\subset\mathbb{R}^3$ is a two-dimensional smooth manifold, and a function $f$ is defined on $S$ such that $f\in C^2(S)$. Let $\vec n=(n_1,~n_2,~n_3)^T$ be the unit outer normal vector of $S$, and $\bar{f}$ be an extension of $f$ in the neighbourhood of $S$ such that $\bar{f}$ is differentiable, then the surface gradient of $f$ on $S$ is defined as
\begin{equation}
\nabla_{_S} f=\nabla\bar{f}-(\nabla \bar{f}\cdot \vec n)\,\vec n,
\end{equation}
with $\nabla$ denoting the usual gradient in $\mathbb{R}^3$. It is easy to show that $\nabla_{_S}f$ is independent of the extension of $f$ and only dependent on the value of $f$ on $S$. If we denote $\nabla_{_S}$ as a vector operator
\begin{equation}
\nabla_{_S} = (\underline{D}_1,~\underline{D}_2,~\underline{D}_3)^T,
\end{equation}
then we can easily obtain
\begin{equation}
\underline{D}_i x_j=\delta_{ij}-n_i\,n_j,\qquad\forall 1\leq i,\,j\leq 3,
\end{equation}
where $\vec x=(x_1,x_2,x_3)$ is the position vector of the surface and $\delta_{ij}$ is the Kronecker delta. The surface divergence of a vector-valued function $\vec g=(g_1,~g_2,~g_3)^T\in [C^1(S)]^3$ is defined as
\begin{equation}
\nabla_{_S}\cdot\vec g = \sum_{i=1}^3\underline{D}_i\,g_i.
\end{equation}
Moreover, the Laplace-Beltrami operator on $S$ can be expressed as
\begin{equation}
\Delta_{_S}\,f=\nabla_{_S}\cdot(\nabla_{_S}\,f) = \sum_{i=1}^3\underline{D}_i \underline{D}_if.
\end{equation}
\end{defi}
In the definition of surface gradient, since the normal component has been subtracted from $\nabla\bar{f}$, $\nabla_{_S}f$ can be viewed as the tangential component of $\nabla\bar{f}$, and thus we have $\nabla_{_S}f \cdot {\vec n}=0$ and
$\nabla_{_S}f(\vec x)\in\mathcal{T}_{\vec x}\,S,\;\forall\vec x\in S$. Note that it can be rigorously proved that Definition~\ref{defi:calculus} is consistent with the conventional definition in differential geometry~\cite{Dziuk13},
and it generalizes the definition domian of surface divergence from the vector in tangent vector spaces to any vector in $\mathbb{R}^3$.

On the other hand, the integration by parts on an open smooth surface $S$ with smooth boundary $\Gamma$ reads as (see Theorem 2.10 in \cite{Dziuk13}, and we omit the proof here)
\begin{equation}
\int_S\nabla_{_S} f\;dS=\int_S f\,\mathcal{H}\,\vec n\;dS+\int_\Gamma f\,\vec c_{_\Gamma}\;d\Gamma,
\label{eqn:InPasurface}
\end{equation}
where $\vec n$ and $\vec c_{_\Gamma}$ are the normal and co-normal vectors (shown in Fig.~\ref{fig:dewetting3d}), respectively, and $\mathcal{H}$ is the mean curvature, which is defined as the surface divergence of the unit normal vector, i.e., $\mathcal{H}=\nabla_{_S}\cdot\vec n$. Similarly, by using the above equation and Definition~\ref{defi:calculus}, we can obtain
the integration by parts about a vector field $\vec F=(f_1,f_2,f_3)^{T} \in \mathbb{R}^3$ defined on an open smooth surface $S$ with smooth boundary $\Gamma$,
\begin{equation}
\int_S\nabla_{_S}\cdot \vec{F}\;dS=\int_S \mathcal{H}\,\vec F \cdot \vec n\;dS+\int_\Gamma \vec F \cdot \vec c_{_\Gamma}\;d\Gamma.
\label{eqn:divergencesurface}
\end{equation}
If $\vec F$ lies in the tangent vector space of $S$, i.e., $\vec F \cdot \vec n=0$, then the first term on the right will vanish.

Furthermore, by using the product rule that $\nabla_{_S}(fg)=g\,\nabla_{_S} f +f\,\nabla_{_S} g$, we can obtain
\begin{equation}
\int_S g\,\nabla_{_S} f \;dS=-\int_S f\,\nabla_{_S} g\;dS+\int_S f\,g\,\mathcal{H}\,\vec n\;dS+\int_\Gamma f\,g\,\vec c_{_\Gamma}\;d\Gamma.
\label{eqn:integralparts}
\end{equation}
In a simple case, if $S$ is a flat surface (i.e., $\mathcal{H}=0$) with a plane boundary curve $\Gamma$, then \eqref{eqn:InPasurface} reduces to
\begin{equation*}
\int_S\nabla_{_S} f\;dS=\int_\Gamma f\,\vec c_{_\Gamma}\;d\Gamma,
\end{equation*}
which is the Gauss-Green theorem in the multivariable calculus, because $\nabla_{_S}f$  collapses to the gradient of $f$ in 2D, and $\vec c_{_\Gamma}$ collapses to the unit outer normal vector of $\Gamma$.

\subsection{The speed method and shape derivative}

In this section, the objective is to calculate the first variation of the energy (or shape) functional defined in \eqref{eqn:totalenergy}. To this end, we first introduce an independent parameter $\epsilon\in[0, \epsilon_0)$
to parameterize a family of perturbations of a given domain $D\subset \mathbb{R}^3$, where the parameter $\epsilon$ controls the amplitude of the perturbation and $\epsilon_0$ is the maximum perturbation amplitude.
Furthermore, we assume that the domain $D$ is of class $C^k$ with $k \geq 2$.

More precisely, we consider a domain $D\subset \mathbb{R}^3$ with a piecewise smooth boundary $\partial D$, then
we can construct a family of transformations $T_\epsilon$, which are one-to-one, and $T_\epsilon$ maps from $\bar{D}$ onto $\bar{D}$,
i.e.,
\begin{equation}
T_\epsilon:\;\bar{D}\;\rightarrow\;\bar{D}, \qquad\epsilon\in[0, \epsilon_0),
\end{equation}
where $\epsilon$ is the small perturbation parameter. Generally, we assume that: (i)~$T_\epsilon$ and $T_\epsilon^{-1}$ belong to $C^k(\bar{D},\mathbb{R}^3)$ for all $\epsilon\in[0,\epsilon_0)$ with $k \geq 2$; and (ii)~the mappings
$\epsilon \rightarrow T_\epsilon(\vec x)$ and $\epsilon \rightarrow T_\epsilon^{-1}(\vec x)$ belongs to $C^1[0,\epsilon_0)$ for all $\vec x\in \mathbb{R}^3$ (with $\vec x=(x_1,x_2,x_3)^T$).

Given any point $\vec X\in\bar{D}$ (with $\vec X=(X_1,X_2,X_3)^T$) and $\epsilon\in[0, \epsilon_0)$, we can define the point $\vec x = T_\epsilon(\vec X)$ which moves along the trajectory. Here, the point $\vec X$ represents the Lagrangian (or material) coordinate, while $\vec x$
is the Eulerian (or actual) coordinate. Therefore, the speed vector field $\vec V(\vec x, \epsilon)$ at point $\vec x$ is defined as
\begin{equation}
\vec V(\vec x,\epsilon)=\frac{\partial \vec x}{\partial \epsilon}(T_\epsilon^{-1}(\vec x),\epsilon).
\end{equation}
On the other hand, the transformation $T_{\epsilon}$ can be uniquely determined by the speed vector field $\vec V$ via the following ordinary differential equation (ODE)
\begin{equation}
\begin{cases}
\frac{d}{d\epsilon}\vec x(\vec X,\epsilon)=\vec V(\vec x(\vec X,\epsilon),\epsilon),\cr
\vec x(\vec X,0)=\vec X.
\end{cases}
\end{equation}
Therefore, the transformation $T_\epsilon$ and the smooth vector field $\vec V$ are uniquely determined by each other.
For a smooth vector field $\vec V$, e.g., $\vec V\in C(C^k(\bar{D},\mathbb{R}^3);[0,\epsilon_0))$, the equivalence between
the transformation $T_\epsilon$ and the speed vector field $\vec V$ has been strictly established by Theorem 2.16 in~\cite{Soko92}. In the following, we use $T_\epsilon(\vec V)$ to denote the transformation associated with vector field $\vec V$. For simplicity, we also denote $\vec V_0=\vec V(\vec X,0)$.

Let $J(G)$ be a shape functional defined on a shape $G \subset \bar{D}$, where $G$ could be a three-dimensional domain (e.g., $\Omega$) or a two-dimensional manifold (e.g., a surface $S$).  The first variation of the functional $J(G)$ at $G$ in the direction of a speed vector field $\vec V\in C(C^k(\bar{D},\bar{D});[0,\epsilon_0))$ is given as the Eulerian derivative:
\begin{equation}
 \delta J(G;\vec V)=\lim_{\epsilon\rightarrow 0}\frac{J(G_\epsilon)-J(G)}{\epsilon},
 \end{equation}
where $G_\epsilon=T_\epsilon(\vec V)(G)$. To obtain the first variation and based on the transformation, we first define the material derivative and shape derivative of a function on a domain $\Omega$ or a surface $S$. For more details about the shape differential calculus, we refer to the book by Sokolowski and Zolesio~\cite{Soko92}.

\begin{defi}\label{defi:material}({\textbf{Material derivatives}}, Def.~2.71 and Def.~2.74 in~\cite{Soko92})
The material derivative $\dot{\psi}(\Omega;\vec V)$ of $\psi$ {\textbf{on a domain $\Omega$}} in the direction of a speed vector field $\vec V$ is defined as
\begin{equation}
 \dot{\psi}(\Omega;\vec V)=\lim_{\epsilon\rightarrow 0}\frac{\psi(\Omega_\epsilon)\circ T_{\epsilon}(\vec V)-\psi(\Omega)}{\epsilon},
\end{equation}
where for $\vec X \in \Omega$, $\psi(\Omega_\epsilon)\circ T_{\epsilon}(\vec V)=\psi(T_{\epsilon}(\vec X))$.

Similarly, the material derivative $\dot{\varphi}(S;\vec V)$ of $\varphi$ {\textbf{on a surface $S$}} in the direction $\vec V$ is defined as
\begin{equation}
 \dot{\varphi}(S;\vec V)=\lim_{\epsilon\rightarrow 0}\frac{\varphi(S_\epsilon)\circ T_{\epsilon}(\vec V)-\varphi(S)}{\epsilon},
\end{equation}
where for $\vec X \in S$, $\varphi(S_\epsilon)\circ T_{\epsilon}(\vec V)=\varphi(T_{\epsilon}(\vec X))$.
\end{defi}

\begin{defi}\label{defi:surfaceshapederivative}
({\textbf{Shape derivatives}}, Def.~2.85 and Def.~2.88 in~\cite{Soko92})
The shape derivative $\psi'(\Omega;\vec V)$ of $\psi$ defined {\textbf{on a domain $\Omega$}}
in the direction $\vec V$ is defined as
\begin{equation}
 \psi'(\Omega;\vec V)=\dot{\psi}(\Omega;\vec V)-\nabla \psi(\Omega)\cdot\vec V_0.
\label{eqn:domainshapeder}
\end{equation}

Similarly, the shape derivative $\varphi'(S;\vec V)$ of $\varphi$ defined {\textbf{on a surface $S$}}
in the direction $\vec V$ is defined as
\begin{equation}
\varphi'(S;\vec V)=\dot{\varphi}(S;\vec V)-\nabla_{_S}\varphi(S)\cdot\vec V_0.
\label{eqn:shapeder}
\end{equation}
\end{defi}

\begin{prop}\label{lm:domainvariation}
Let $\Omega\subset \mathbb{R}^3$ be a smooth bounded domain in $\bar{D}$ with smooth boundary $\partial \Omega$, and $\vec V$ be a speed vector field such that $\vec V\in C(C^k(\bar{D},\bar{D});[0,\epsilon_0))$. Suppose that $\psi=\psi(\Omega)$ is given such that the material derivative $\dot{\psi}(\Omega;\vec V)$ and the shape derivative $\psi'(\Omega;\vec V)$ exist.
Then, the shape functional $J(\Omega)=\int_\Omega \psi(\Omega)\,d\Omega$ is shape differentiable and
we have
\begin{equation}
\delta J(\Omega; \vec V)=\int_\Omega \psi'(\Omega;\vec V)\;d\Omega + \int_{\partial\Omega}\psi(\Omega)\,\vec V_0\cdot\vec n\;d\Omega.
\label{eqn:domainvariation}
\end{equation}
\end{prop}

\begin{proof}
See Section 2.31 on Pages 112-113 in~\cite{Soko92}.
\end{proof}

\begin{remark}
If~$\vec V_0\cdot\vec n=0$ on the boundary $\partial\Omega$, we obtain that the first variation of the functional reduces to
\begin{equation}
\delta J(\Omega;\vec V) = \int_\Omega \psi'(\Omega;\vec V)\;d\Omega.
\end{equation}
\end{remark}

Furthermore, the definition of shape derivative for a function $\varphi(S)$ defined over a two-dimensional manifold $S$, ensures that the shape derivative shows no dependence on the extension of $\varphi$ in the near neighbourhood. We propose the following proposition to show that the first variation of a functional on $S$ is closely related to the shape derivative.

\begin{prop}\label{lm:surfaceintegralvariation}
Let $S$ be a two-dimensional smooth manifold in $\bar{D}$ with smooth boundary $\Gamma$, and $\vec V$ be a speed vector field such that $\vec V\in C(C^k(\bar{D},\bar{D});[0,\epsilon_0))$. Suppose that $\varphi=\varphi(S)$ is given such that the material derivative $\dot{\varphi}(S;\vec V)$ and the shape derivative $\varphi'(S;\vec V)$ exist.
Then, the shape functional $J(S)=\int_S \varphi(S)\,dS$ is shape differentiable and we have
 \begin{equation}
 \delta J(S;\vec V)=\int_S\varphi'(S;\vec V)\;dS+\int_S\varphi(S)\mathcal{H}\vec V_0\cdot\vec n\;dS+\int_\Gamma\varphi(S)\vec V_0\cdot\vec c_{_\Gamma}\;d\Gamma,
 \label{eqn:shapelemma}
 \end{equation}
 where $\mathcal{H}$ is the mean curvature of the surface $S$, and $\vec c_{_\Gamma}$ is the unit co-normal vector.
 Furthermore, if $\varphi(S)=\psi(\Omega)\Big|_{S}$, then we have
 \begin{equation}
\delta J(S;\vec V)=\int_S \psi'(\Omega;\vec V)\Big|_S\;dS+\int_S\Bigl(\frac{\partial \psi}{\partial \vec n}+ \psi \mathcal{H}\Bigr)\vec V_0\cdot\vec n\;dS+\int_\Gamma \psi\vec V_0\cdot\vec c_{_\Gamma}\;d\Gamma.
\label{eqn:shapefunction}
\end{equation}
 \end{prop}

\begin{proof}
By referring to Section 2.33 on Pages 115-116 in~\cite{Soko92}, we can directly obtain
\begin{equation}
\delta J(S;\vec V)
=\int_S \dot{\varphi}(S;\vec V)\;dS +\int_S\varphi(S)\nabla_{_S}\cdot\vec V_0\;dS.
\label{eqn:firststepshape}
\end{equation}
By using integration by parts and also making use of \eqref{eqn:shapeder}, we obtain
\begin{eqnarray}
\delta J(S;\vec V)&=&\int_S\dot{\varphi}(S;\vec V_0)\;dS-\int_S \nabla_{_S}\varphi(S)\cdot\vec V_0\;dS+\int_S\varphi(S)\mathcal{H}\vec V_0\cdot \vec n \;dS\nonumber\\
&&+\int_\Gamma\varphi(S)\vec V_0\cdot\vec c_{_\Gamma}\;d\Gamma\nonumber\\
&=&\int_S\varphi'(S;\vec V)\;dS+\int_S\varphi(S)\mathcal{H}\vec V_0\cdot\vec n\;dS+\;\int_\Gamma\varphi(S)\vec V_0\cdot\vec c_{_\Gamma}\;d\Gamma. \label{eqn:shapelemma2}
\end{eqnarray}

Furthermore, notice that the definitions of shape derivatives on the domain $\Omega$ and the surface $S$ are different
(see \eqref{eqn:domainshapeder}-\eqref{eqn:shapeder}). If we assume that $\psi$ is a function defined on the domain $\Omega$, such that its restriction on $S$ equals to the function $\varphi(S)$, i.e., $\psi(\Omega)\Big|_{S} = \varphi(S)$, we can reformulate \eqref{eqn:shapelemma2} in terms of the extension function $\psi(\Omega)$ as
\begin{equation*}
\delta J(S;\vec V)=\int_S \psi'(\Omega;\vec V)\Big|_S\;dS+\int_S\Bigl(\frac{\partial \psi}{\partial \vec n}+ \psi \mathcal{H}\Bigr)\vec V_0\cdot\vec n\;dS+\int_\Gamma \psi\vec V_0\cdot\vec c_{_\Gamma}\;d\Gamma,
\end{equation*}
which completes the proof.
\end{proof}

\vspace{0.1cm}
\begin{remark}
If $S$ is a closed surface, then the boundary term about $\Gamma$ in~\eqref{eqn:shapefunction} will vanish.
The similar results for a closed curve or surface can be found in~\cite{Dougan12,Hintermuller2004}.
\end{remark}

In the following, we will apply~\eqref{eqn:shapefunction} in Proposition~\ref{lm:surfaceintegralvariation} for calculating the first variation of the energy (or shape) functional defined in \eqref{eqn:totalenergy}, where the integrand is the surface energy density $\gamma(\vec n)$. To calculate the shape derivatives and obtain the first variation, we shall make use of the signed distance function, which is a powerful tool in shape sensitivity analysis. Consider a closed domain $\Omega \subset \mathbb{R}^3$ with a smooth boundary surface $\partial\Omega$, and then the signed distance function is defined as
\begin{equation}
\label{eqn:signdistance}
b(\vec x) =
\begin{cases}
\rm{dist}(\vec x,\partial\Omega),\quad&\forall\vec x\in \mathbb{R}^3\backslash\Omega,
\cr
0,\qquad&\forall\vec x\in \partial\Omega,\cr
-\rm{dist}(\vec x,\partial\Omega),\quad&\forall\vec x\in\Omega.
\end{cases}
\end{equation}
Here, $\rm{dist}(\vec x,\partial\Omega)=\inf_{\vec y\in\partial\Omega}||\vec x-\vec y||$. The signed distance function $b(\vec x)$ can be used to determine the unit outer normal vector $\vec n$ and the mean curvature $\mathcal{H}$ on the boundary surface $\partial\Omega$. More precisely, we can extend the functions $\vec n$ and $\mathcal{H}$ which are defined on $\partial\Omega$  in terms of $b(\vec x)$ in a tubular neighbourhood such that
\begin{equation}
\label{eqn:normcurvature}
\vec n(\vec x) = \nabla b(\vec x)\Big|_{\partial\Omega},\qquad\mathcal{H}(\vec x) = \Delta b(\vec x)\Big|_{\partial\Omega},\qquad\forall\vec x\in\partial\Omega.
\end{equation}
The shape derivative of the signed distance function in the direction of a vector field $\vec V$ is calculated as $b'(\Omega;\vec V)=-\vec V_0\cdot\vec n$ (see~\cite{Dougan12,Hintermuller2004} for more details). Moreover, based on the extension, the shape derivatives of the two extension functions restricted on $\partial\Omega$ are also obtained (see Lemma~3.1 in~\cite{Dougan12}), i.e.,
\begin{equation}
\vec n'(\Omega;\vec V)\Big|_{\partial \Omega}=-\nabla_{_S}(\vec V_0\cdot\vec n),\qquad \mathcal{H}'(\Omega;\vec V)\Big|_{\partial \Omega}=-\Delta_S(\vec V_0\cdot\vec n).
\label{eqn:shapenormal}
\end{equation}

\subsection{First variation}

By applying~\eqref{eqn:shapefunction} and making use of the shape derivative of the unit outer normal vector, we obtain the following lemma.

\begin{lem}\label{lem:variationsurfaceenergies}
Assume that $S \subset \bar{D}$ is a two-dimensional smooth manifold with smooth boundary $\Gamma$. Let $\vec n$ be the unit outer normal vector of $S$, and  $\vec V$ be a speed vector field such that $\vec V\in C(C^k(\bar{D},\bar{D});[0,\epsilon_0))$. If the shape functional $J(S)=\int_S\gamma(\vec n)\;dS$ with a surface energy (density) $\gamma(\vec n)$, then the first variation of $J(S)$ is given as
\begin{equation}
\delta J(S;\vec V) = \int_S\,(\nabla_{_S}\cdot\boldsymbol{\xi})\,(\vec V_0\cdot\vec n)\,dS + \int_\Gamma\,\vec V_0 \cdot \vec c_{_\Gamma}^\gamma\,d\Gamma,
\label{eqn:Wvariation}
\end{equation}
where $\boldsymbol{\xi}:=\boldsymbol{\xi}(\vec n)$ is the Cahn-Hoffman vector, which is defined previously in \eqref{eqn:xivector}, and $\vec V_0\cdot \vec n$ represents the deformation velocity along the outer normal direction of the interface $S$, and the vector $\vec c_{_\Gamma}^\gamma:=(\boldsymbol{\xi}\cdot\vec n)\,\vec c_{_\Gamma}-(\boldsymbol{\xi}\cdot\vec c_{_\Gamma})\,\vec n$ with $\vec c_{_\Gamma}$ representing the unit co-normal vector
(shown in Fig.~\ref{fig:dewetting3d}).
\end{lem}

\begin{proof}
We firstly assume $\hat{\gamma}(\vec p)$ is a homogeneous extension of $\gamma(\vec n)$,
\begin{equation}
\hat{\gamma}(\vec p) = |\vec p|\gamma\left(\frac{\vec p}{|\vec p|}\right),\quad\forall\, \vec p\in\mathbb{R}^3\backslash\{{\vec 0}\},
\end{equation}
where the definition domain of the function $\gamma(\vec n)$ changes from unit vectors $\vec n$ to arbitrary non-zero vectors $\vec p\in\mathbb{R}^3$.

We next consider a bounded domain $\Omega\subset \mathbb{R}^3$ such that $S\subset\partial\Omega$. Then, based on the signed distance function defined in \eqref{eqn:signdistance}, we can define $\nabla b(x)\in\mathbb{R}^3$ as an extension of the normal vector $\vec n$ in the neighbourhood of $S$.  Thus we can reformulate
\begin{equation}
J(S) = \int _S \hat{\gamma}\bigl(\nabla b(\vec x)\bigr)\Big|_S \;dS = \int_S \psi(\Omega)\Big|_S\;dS,
\end{equation}
with $\psi(\Omega): = \hat{\gamma}\bigl(\nabla b(\vec x)\bigr)$. Using the chain rule for shape derivatives and the definition of Cahn-Hoffman $\boldsymbol{\xi}$-vector in~\eqref{eqn:xivector}, we conclude that the following expression holds
\begin{equation}
\psi'(\Omega;\vec V)\Big|_S = \nabla{\hat{\gamma}}\bigl(\nabla b(\vec x)\bigr)\Big|_S \cdot \vec n'(\Omega;\vec V)\Big|_S = -\boldsymbol{\xi}\cdot\nabla_{_S}(\vec V_0\cdot\vec n).
\label{eqn:shapevariation1}
\end{equation}
Moreover, by noting the fact $|\nabla b(\vec x)|= 1$, we obtain
\begin{equation}
\frac{\partial \psi}{\partial \vec n} \Big|_S = \boldsymbol{\xi}\cdot\bigl(\left(\nabla\nabla b(\vec x)\right)\, \nabla b(\vec x)\bigr)\Big|_S= 0,
\label{eqn:shapevariation2}
\end{equation}
where $\nabla\nabla b(\vec x)\in \mathbb{R}^{3\times 3}$.
By making use of \eqref{eqn:shapefunction} and combining \eqref{eqn:shapevariation1}-\eqref{eqn:shapevariation2}, we immediately have
\begin{eqnarray}
\delta J(S;\vec V) &=& - \int_S \boldsymbol{\xi}\cdot\nabla_{_S}(\vec V_0\cdot\vec n)\;dS + \int_S\gamma(\vec n)\,(\vec V_0\cdot\vec n)\,\mathcal{H}\;dS +\int_\Gamma \gamma(\vec n)\,(\vec V_0\cdot\vec c_{_\Gamma})\;d\Gamma\nonumber\\
&:=& I + II + III.   \nonumber
\end{eqnarray}
For the first term, by using the integration by parts, we obtain
\begin{equation*}
I = \int_S(\nabla_{_S}\cdot\boldsymbol{\xi})\,(\vec V_0\cdot\vec n)\;dS-\int_S(\boldsymbol{\xi}\cdot\vec n)\,(\vec V_0\cdot \vec n)\,\mathcal{H}\;dS -\;\int_\Gamma(\boldsymbol{\xi}\cdot\vec c_{_\Gamma})\,(\vec V_0\cdot \vec n)\;d\Gamma.
\end{equation*}
Based on~\eqref{eqn:xiequation}, we have $\gamma(\vec n) = \boldsymbol{\xi}\cdot\vec n$. Thus we can rewrite
\begin{equation*}
II = \int_S(\boldsymbol{\xi}\cdot\vec n)\,(\vec V_0\cdot\vec n)\mathcal{H}\;dS, \qquad
III = \int_\Gamma(\boldsymbol{\xi}\cdot\vec n)\,(\vec V_0\cdot\vec c_{_\Gamma})\;d\Gamma.
\end{equation*}
Finally, by combining the above three terms together, we immediately have
\begin{eqnarray}
\delta J(S;\vec V)
&=&\int_S(\nabla_{_S}\cdot\boldsymbol{\xi})\,(\vec V_0\cdot\vec n)\;dS+\int_\Gamma\,\Bigl[(\boldsymbol{\xi}\cdot\vec n)\,\vec c_{_\Gamma}-(\boldsymbol{\xi}\cdot\vec c_{_\Gamma})\,\vec n\Bigr]\cdot\vec V_0\;d\Gamma\nonumber\\
&=&\int_S(\nabla_{_S}\cdot\boldsymbol{\xi})\;(\vec V_0\cdot\vec n)\;dS+\int_\Gamma \vec c_{_\Gamma}^\gamma\cdot\vec V_0\;d\Gamma, \nonumber
\end{eqnarray}
where $\vec c_{_\Gamma}^\gamma=(\boldsymbol{\xi}\cdot\vec n)\,\vec c_{_\Gamma}-(\boldsymbol{\xi}\cdot\vec c_{_\Gamma})\,\vec n$.
\end{proof}

By using the above Lemma, we can easily obtain the first variation of the energy functional for solid-state dewetting problems
defined in \eqref{eqn:totalenergy}.
\begin{thm}
\label{thm:Wvariation}
The first variation of the free energy (or shape) functional \eqref{eqn:totalenergy} used in solid-state dewetting problems
with respect to a smooth vector field $\vec V$ can be written as:
\begin{equation}
\delta W(S;\vec V) = \int_S\,(\nabla_{_S}\cdot\boldsymbol{\xi})\;(\vec V_0\cdot\vec n)\,dS + \int_\Gamma(\vec c_{_\Gamma}^\gamma\cdot\vec n_{_\Gamma} + \gamma_{_{\subFS}} - \gamma_{_{\subVS}})(\vec V_0\cdot\vec n_{_\Gamma})\,d\Gamma,
\label{eqn:energyvariation}
\end{equation}
where $\vec n_{_\Gamma}$ is the unit outer normal of the contact line curve $\Gamma$ on the substrate (shown in Fig.~\ref{fig:dewetting3d}).
\end{thm}

\begin{proof}
From \eqref{eqn:totalenergy}, we observe that the total free energy consists of two parts: the film/vapor interface energy
$W_{\rm{int}}$ and the substrate energy $W_{\rm{sub}}$. First, by using Lemma~\ref{lem:variationsurfaceenergies}, we
can directly obtain the first variation of the film/vapor interface energy
$W_{\rm{int}}$ as follows,
\begin{equation}
\delta W_{\rm{int}}(S;\vec V) = \int_S(\nabla_{_S}\cdot\boldsymbol{\xi})\,(\vec V_0\cdot\vec n)\;dS+\int_\Gamma \vec V_0 \cdot \vec c_{_\Gamma}^\gamma\;d\Gamma.
\label{eqn:tv1}
\end{equation}
Here, $\vec c_{_\Gamma}^\gamma$ is a linear combination of $\vec c_{_\Gamma}$ and $\vec n$, which is defined on the contact line $\Gamma$. Therefore, as shown in Fig.~\ref{fig:dewetting3d}, we have
\begin{equation}
\vec c_{_\Gamma}\perp\boldsymbol{\tau}_{_\Gamma},\quad\vec n\perp\boldsymbol{\tau}_{_\Gamma}\quad\Rightarrow\quad \vec c_{_\Gamma}^\gamma\perp\boldsymbol{\tau}_{_\Gamma}.
\label{eqn:ccrelation}
\end{equation}
For solid-state dewetting problems studied in this paper, we assume that the contact line $\Gamma$ must move along the substrate plane $S_{\rm{sub}}$, i.e.,
\begin{equation*}
T_\epsilon\Gamma\subset S_{\rm{sub}},\qquad \vec V_0(\vec x) \sslash S_{\rm{sub}},\qquad\forall\vec x\in\Gamma.
\end{equation*}
Therefore, for any $\vec x \in \Gamma$, $\vec V_0(\vec x)$ can be decomposed into two vectors along the directions $\vec n_{_\Gamma}(\vec x)$ and $\boldsymbol{\tau}_{_\Gamma}(\vec x)$, i.e., $\vec V_0=k_1\vec n_{_\Gamma}+k_2\boldsymbol{\tau}_{_\Gamma}$,
where $k_1$ and $k_2$ represent the corresponding components. By making use of \eqref{eqn:ccrelation}, we can obtain
\begin{eqnarray*}
\vec V_0 \cdot  \vec c_{_\Gamma}^\gamma &=& (k_1\vec n_{_\Gamma}+k_2\boldsymbol{\tau}_{_\Gamma})\cdot  \vec c_{_\Gamma}^\gamma = k_1(\vec n_{_\Gamma} \cdot  \vec c_{_\Gamma}^\gamma)  \\
&=&(\vec V_0 \cdot \vec n_{_\Gamma})\,(c_{_\Gamma}^\gamma\cdot\vec n_{_\Gamma}), \qquad\forall\vec x\in\Gamma.
\end{eqnarray*}
Thus we can reformulate~\eqref{eqn:tv1} as
\begin{equation}
\delta W_{\rm{int}}(S;\vec V) = \int_S(\nabla_{_S}\cdot\boldsymbol{\xi})\;(\vec V_0\cdot\vec n)\,dS + \int_\Gamma (\vec c_{_\Gamma}^\gamma\cdot\vec n_{_\Gamma})(\vec V_0 \cdot \vec n_{_\Gamma})\,d\Gamma.
\label{eqn:WIvariation}
\end{equation}

On the other hand, we can rewrite the substrate energy $W_{\rm{sub}}$ as
\begin{equation*}
W_{\rm{sub}} = (\gamma_{_{\subFS}} - \gamma_{_{\subVS}})A(\Gamma) = (\gamma_{_{\subFS}} - \gamma_{_{\subVS}})\int_{S_{_{\subFS}}}\;d S_{_{\subFS}}.
\end{equation*}
By using Proposition~\ref{lm:surfaceintegralvariation}, and noting that the integrand $\varphi$ in~\eqref{eqn:shapelemma} is a constant and $S_{_{\subFS}}$ is a flat surface with a plane boundary curve $\Gamma$ (i.e., $\mathcal{H}=0$ and $\vec n_{_\Gamma}$ is the unit co-normal vector of the flat surface $S_{_{\subFS}}$), we directly have
\begin{equation}
\delta W_{\rm{sub}}(S;\vec V) = (\gamma_{\subFS} - \gamma_{\subVS})\int_\Gamma \vec V_0\cdot\vec n_{_\Gamma}\;d\Gamma.
\label{eqn:WSvariation}
\end{equation}
By combining~\eqref{eqn:WIvariation} and \eqref{eqn:WSvariation}, we obtain the following conclusion
\begin{equation*}
\delta W(S;\vec V) = \int_S\,(\nabla_{_S}\cdot\boldsymbol{\xi})\;(\vec V_0\cdot\vec n)\,dS + \int_\Gamma(\vec c_{_\Gamma}^\gamma\cdot\vec n_{_\Gamma} + \gamma_{\subFS}-\gamma_{\subVS})(\vec V_0\cdot\vec n_{_\Gamma})\,d\Gamma,
\end{equation*}
which completes the proof.
\end{proof}

\begin{remark}
The variational result given by~\eqref{eqn:energyvariation} tells us that the rate of change of the total interfacial
free energy is contributed from the two parts: one part results from the change of the interface $S$, and it is proportional
to the weighted mean curvature (i.e., $\nabla_{_S}\cdot\boldsymbol{\xi}$)~\cite{Taylor92} and the rate of change of the volume (i.e., $\vec V_0\cdot\vec n\,dS$, the normal velocity times the surface area element); the other part comes from the change of the contact line $\Gamma$.
\end{remark}

\begin{remark}
In 2D case, the variational result given by~\eqref{eqn:energyvariation} in Theorem~\ref{thm:Wvariation} will reduce to the
variational result presented in the reference~\cite{Jiang18}.
\end{remark}

\begin{remark}
When the substrate is curved in 3D, the variational result given by~\eqref{eqn:energyvariation} in Theorem~\ref{thm:Wvariation} is still valid. We can perform similar discussions as the reference~\cite{Jiang18b} for curved substrates in 2D.
\end{remark}

\section{Equilibrium shapes}

The equilibrium shape of the solid-state dewetting problem can be stated as follows~\cite{Bao17b,Jiang16}:
\begin{equation}
\min_{\Omega} W : =W(S) = \int_S\gamma(\vec n)\;dS + (\gamma_{_{\subFS}}-\gamma_{_{\subVS}})A(\Gamma) \quad\rm{s.t.}\quad|\Omega|=C,
\label{eqn:3dewetting}
\end{equation}
where $C>0$ is a prescribed constant representing the total volume of the dewetted particle, and $\Omega$ represents the domain (or the particle) enclosed by the interface $S$ and the substrate plane $S_{\rm{sub}}$.

The Lagrangian for the above optimization problem can be defined as
\begin{equation}
L(S,\lambda) = \int_S\gamma(\vec n)\;dS + (\gamma_{_{\subFS}}-\gamma_{_{\subVS}})A(\Gamma) - \lambda (|\Omega|-C),
\end{equation}
with $\lambda$ representing the Lagrange multiplier. The first variation of the total volume term can be obtained by simply choosing the integrand $\psi(\vec x) \equiv 1, \forall\,\vec x \in \Omega$ in~\eqref{eqn:domainvariation} by Proposition~\ref{lm:domainvariation}. Therefore, by combining with~\eqref{eqn:energyvariation}, the first variation of the Lagrangian with respect to a smooth vector field $\vec V$ can be given as
\begin{equation}
\delta L(S,\lambda;\vec V) = \int_S(\nabla_{_S}\cdot\boldsymbol{\xi} - \lambda)(\vec V_0\cdot\vec n)\;dS + \int_\Gamma(\vec c_{_\Gamma}^\gamma\cdot\vec n_{_\Gamma} + \gamma_{_{\subFS}} - \gamma_{_{\subVS}})(\vec V_0\cdot\vec n_{_\Gamma})\,d\Gamma.
\label{eqn:LagrangianVariation}
\end{equation}
Based on the above first variation, we have the following theorem which yields the necessary conditions for the equilibrium shape of solid-state dewetting problem.
\begin{theorem}
\label{lem:equcondition}
Assume that a two-dimensional manifold $S_e$ with smooth boundary $\Gamma_e$ is the equilibrium shape of the solid-state dewetting problem~\eqref{eqn:3dewetting}, then the following conditions must be satisfied
\begin{subequations}
\begin{align}
\label{eqn:equcondition1}
&\nabla_{_{S_e}}\cdot\boldsymbol{\xi}=\lambda,\;\text{on}\;S_e.\\
&\vec c_{_\Gamma}^\gamma\cdot\vec n_{_\Gamma} + \gamma_{_{\subFS}} - \gamma_{_{\subVS}}=0,\;\text{on}\;\Gamma_e.
\label{eqn:equcondition2}
\end{align}
\end{subequations}
where the constant $\lambda$ is determined by the prescribed total volume, i.e., the constant $C$.
\begin{proof}
If $S_e$ is the equilibrium shape, then~\eqref{eqn:LagrangianVariation} must vanish at $S=S_e$ for any smooth vector field $\vec V$. Therefore, we immediately obtain the above two necessary conditions.
\end{proof}
\end{theorem}

\begin{figure}[!htp]
\centering
\includegraphics[width=0.45\textwidth]{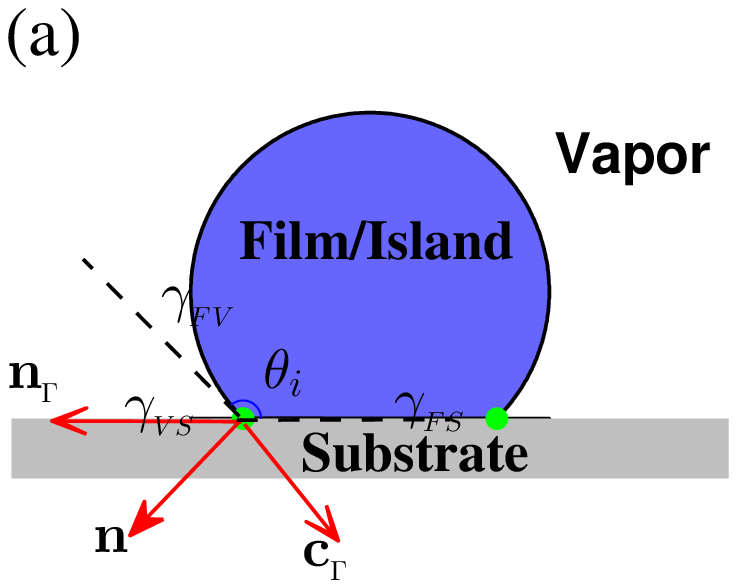}
\includegraphics[width=0.45\textwidth]{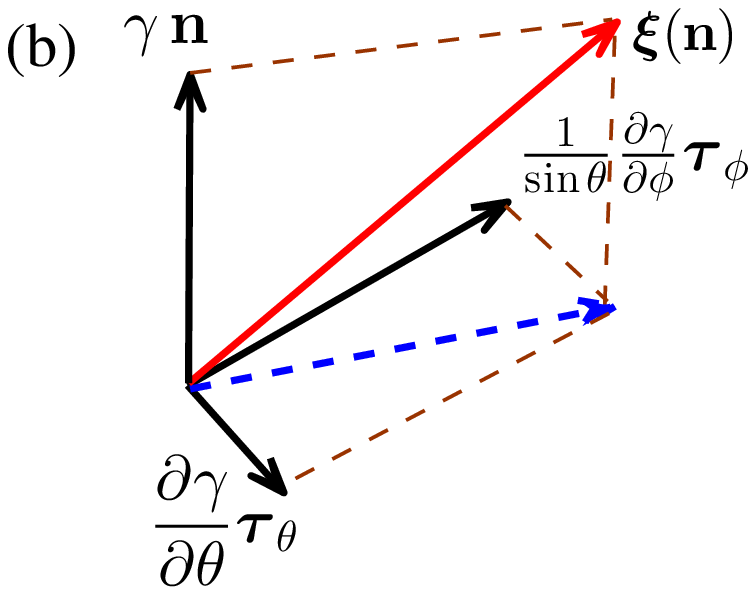}
\caption{(a) The cross-section profile of the equilibrium shape associated with several vectors at the contact line $\Gamma$, where $\theta_i$ is the equilibrium contact angle; (b) the three components of the Cahn-Hoffman $\boldsymbol{\xi}$-vector in the spherical coordinate system.}
\label{fig:Equvector}
\end{figure}

For isotropic surface energy, i.e., $\gamma(\vec n)\equiv 1$ (scaled by a constant $\gamma_0$), we have $\boldsymbol{\xi}=\vec n$ and $\vec c_{_\Gamma}^\gamma=\vec c_{_\Gamma}$. By simple calculations, \eqref{eqn:equcondition1} will reduce to the condition of constant mean curvature. Denote $\Gamma_e$ as the boundary of $S_e$. For arbitrary $\vec x\in\Gamma_e$, let $\theta_i(\vec x)$ represent the equilibrium contact angle at boundary point $\vec x$.  Then, \eqref{eqn:equcondition2} will reduce to
\begin{equation}
\cos\theta_i(\vec x) = \sigma,\quad\forall\vec x\in\Gamma_e,
\label{eqn:Youngequation}
\end{equation}
where the (dimensionless) material constant $\sigma: = \frac{\gamma_{_{\subVS}}-\gamma_{_{\subFS}}}{\gamma_{0}}=\cos\theta_i$,
and it is the well-known isotropic Young equation \cite{Young1805}.

Condition~\eqref{eqn:equcondition2} can be regarded as the Young equation for anisotropic surface energy $\gamma(\vec n)$ in 3D. For the anisotropic case, we can write the surface energy density in terms of the spherical coordinate, i.e., $\gamma_{_{\subFV}}=\gamma(\theta,\phi)$ (scaled by a constant $\gamma_0$). Therefore, the Cahn-Hoffman $\boldsymbol{\xi}$-vector can be decomposed into the following three components (as shown in Fig.~\ref{fig:Equvector}(b)):
\begin{equation}
\boldsymbol{\xi}(\vec n) = \nabla\hat{\gamma}(\vec n) = \gamma(\theta,\phi)\vec n  + \frac{\partial\gamma(\theta,\phi)}{\partial\theta}\boldsymbol{\tau}_{_\theta} + \frac{1}{\sin\theta}\frac{\partial\gamma(\theta,\phi)}{\partial\phi}\boldsymbol{\tau}_{_\phi},
\label{eqn:cahnhoff2}
\end{equation}
where in these expressions,
\begin{subequations}
\begin{align}
&\vec n = (\sin\theta\cos\phi,\sin\theta\sin\phi,\cos\theta)^T, \nonumber\\
&\boldsymbol{\tau}_{_\theta} = (\cos\theta\cos\phi,\cos\theta\sin\phi,-\sin\theta)^T,\nonumber\\
&\boldsymbol{\tau}_{_\phi} = (-\sin\phi,\cos\phi,0)^T. \nonumber
\end{align}
\end{subequations}
We obtain that the following expressions hold:
\begin{subequations}
\begin{align}
&\boldsymbol{\xi}\cdot\vec n = \gamma(\theta,\phi),\qquad \vec c_{_\Gamma}\cdot\vec n_{_\Gamma} = \cos\theta(\vec x), \nonumber\\
&\boldsymbol{\xi}\cdot\vec c_{_\Gamma} = \frac{\partial\gamma(\theta,\phi)}{\partial\theta},\qquad\vec n\cdot\vec n_{_\Gamma} = \sin\theta(\vec x). \nonumber
\end{align}
\end{subequations}
Therefore, we can rewrite~\eqref{eqn:equcondition2} as
\begin{equation}
\gamma(\theta,\phi)\cos \theta(\vec x) - \frac{\partial\gamma(\theta,\phi)}{\partial\theta}\sin \theta(\vec x)-\sigma=0,
\quad\forall\vec x\in\Gamma_e,
\end{equation}
which is consistent with the anisotropic Young equation discussed for the solid-state dewetting problem in 2D \cite{Bao17b,Wang15}.

If $\vec X:=\vec X(\theta,\phi)$ represents the position vector of a surface, we have $\nabla_{_S}\cdot\vec X=2$ by using Definition~\ref{defi:calculus}. Therefore, if we use the $\boldsymbol{\xi}$-plot to represent the position vector of the equilibrium shape, then the necessary condition~\eqref{eqn:equcondition1} will be automatically satisfied. From one side, this is the reason why the $\boldsymbol{\xi}$-plot can yield equilibrium shapes for a free-standing solid particles (as shown in Fig~\ref{fig:anisotropy}). Furthermore, based on the recent work for the generalized Winterbottom construction~\cite{Bao17b,Winterbottom67}, we can construct its analytical expression for the equilibrium shape which also can satisfy the contact angle condition~\eqref{eqn:equcondition2}. First, we define a domain of definition $U_{_\phi}$ for $\theta$ under a fixed value $\phi$ as
\begin{equation}
U_{_\phi}:=\Bigl\{\theta\Big|\gamma(\theta,\phi)\cos \theta - \frac{\partial\gamma(\theta,\phi)}{\partial\theta}\sin\theta-\sigma\geq0, \quad\theta\in [0, \pi]\Bigr\},
\end{equation}
where $\sigma = \frac{\gamma_{_{_{\subVS}}}-\gamma_{_{_{\subFS}}}}{\gamma_{_0}}$. Based on Theorem~\ref{lem:equcondition}, we can explicitly construct its equilibrium shape in the parametric formula as $S_e(\theta,\phi):=\vec X(\theta,\phi)=(x(\theta,\phi),y(\theta,\phi),z(\theta,\phi))^T$,
\begin{equation}
\begin{cases}
x(\theta,\phi)= \lambda\bigl[\gamma(\theta,\phi)\sin\theta\cos\phi + \frac{\partial\gamma(\theta,\phi)}{\partial\theta}\cos\theta\cos\phi - \frac{1}{\sin\theta}\frac{\partial\gamma(\theta,\phi)}{\partial\phi}\sin\phi\bigr],\\[0.5em]
y(\theta,\phi)=\lambda\bigl[\gamma(\theta,\phi)\sin\theta\sin\phi + \frac{\partial\gamma(\theta,\phi)}{\partial\theta}\cos\theta\sin\phi + \frac{1}{\sin\theta}\frac{\partial\gamma(\theta,\phi)}{\partial\phi}\cos\phi\bigr],\\[0.5em]
z(\theta,\phi)=\lambda\bigl[\gamma(\theta,\phi)\cos\theta - \frac{\partial\gamma(\theta,\phi)}{\partial\theta}\sin\theta -\sigma\bigr],
\end{cases}
\label{eqn:equilibirumshape}
\end{equation}
where $\phi\in[0,2\pi]$, $\theta\in U_{_\phi}$, and $\lambda$ is the scaling constant determined by the total volume $|\Omega|$.

\begin{figure}[!htp]
\centering
\includegraphics[width=1.0\textwidth]{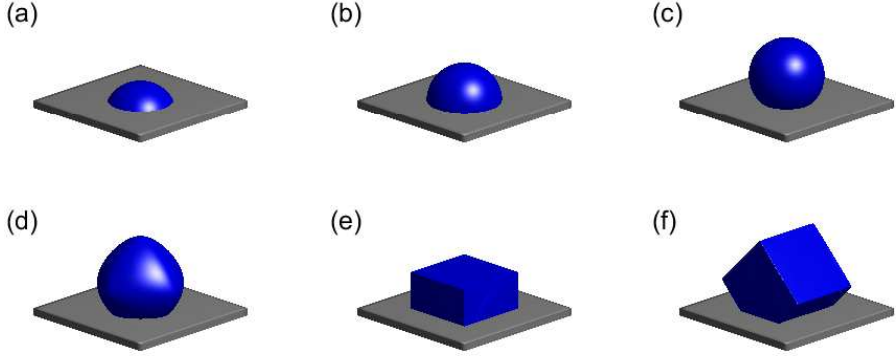}
\caption{The equilibrium shape defined by Eq.~\eqref{eqn:equilibirumshape}, where (a)-(c) is for isotropic surface energy, i.e., $\gamma(\vec n) \equiv 1$, but with different material constants $\sigma=\cos(\pi/3), \cos(\pi/2), \cos(3\pi/4)$, respectively; (d) $\gamma(\vec n) = 1+ 0.2(n_1^4+n_2^4+n_3^4)$,~$\sigma=\cos(3\pi/4)$; (e) The surface energy density is given by~\eqref{eqn:smoothcusp}, where~$\sigma=\cos(3\pi/4),~\varepsilon=0.01$ ; (f) The surface energy density is given by $\gamma(\boldsymbol{M}_x(\pi/4)\vec n)$ where $\gamma(\vec n)$ is defined by~\eqref{eqn:smoothcusp} and $\boldsymbol{M}_x(\pi/4)$ represents an orthogonal matrix for the rotation by an angle $\pi/4$ about the $x$-axis in 3D, using the right-hand rule, where~$\sigma=\cos(3\pi/4)$,~$\varepsilon=0.01$.}
\label{fig:Equilibirum}
\end{figure}

Based on the formula~\eqref{eqn:equilibirumshape}, the equilibrium shape under different types of surface energy anisotropies, e.g., the cubic anisotropy and regularized ``cusped'' anisotropy defined in~\eqref{eqn:smoothcusp}, can be easily constructed. Fig.~\ref{fig:Equilibirum}(a)-(c) depicts the equilibrium shapes for isotropic surface energy with the material constant $\sigma$ chosen as $\sigma=\cos(\pi/3),\cos(\pi/2),\cos(3\pi/4)$, respectively. It clearly demonstrates the effect of the material constant $\sigma$ on the equilibrium shape by influencing the equilibrium contact angle via \eqref{eqn:Youngequation}. Moreover, we also present equilibrium shapes for the cubic anisotropic surface energy, i.e., $\gamma(\vec n) = 1 + a(n_1^4+n_2^4+n_3^4)$ and regularized ``cusped'' surface energy defined in~\eqref{eqn:smoothcusp} with $\sigma=\cos(3\pi/4)$ in Fig.~\ref{fig:Equilibirum}(d)-(e). The anisotropy for Fig.~\ref{fig:Equilibirum}(f) is chosen by an anti-clockwise rotation along the $x$-axis by $45$ degrees under the right-hand rule for the regularized ``cusped'' surface energy. We can observe that this rotation results in a corresponding rotation of the equilibrium shape.

\section{A sharp-interface model and its properties}

In this section, we propose a kinetic sharp-interface model for simulating solid-state dewetting of thin films with anisotropic surface energies, and then we show that the proposed model satisfies the mass conservation and energy dissipation.
\subsection{The model}
Based on~\eqref{eqn:energyvariation} in Theorem~\ref{thm:Wvariation}, we can define the first variation of the total interfacial energy functional with respect to the film/vapor interface $S$ and its boundary curve (i.e., the contact line $\Gamma$) as
\begin{equation}
\frac{\delta W}{\delta S}=\nabla_{_S}\cdot\boldsymbol{\xi},\qquad \frac{\delta W}{\delta \Gamma}=\vec c_{_\Gamma}^\gamma\cdot\vec n_{_\Gamma}+\gamma_{_{\subFS}} - \gamma_{_{\subVS}}.
\end{equation}
From the Gibbs-Thomson relation~\cite{Mullins57,Sutton95}, the chemical potential can be defined as
\begin{equation}
\mu=\Omega_0\frac{\delta W}{\delta S}=\Omega_0\nabla_{_S}\cdot\boldsymbol{\xi},
\label{eqn:svn1}
\end{equation}
with $\Omega_0$ representing the atomic volume. The normal velocity of the moving interface is controlled by surface diffusion \cite{Cahn74,Mullins57,Wang15,Jiang16}, and it can be defined as follows by Fick's laws of diffusion~\cite{Balluffi05}
\begin{equation}
\vec J = -\frac{D_s\nu}{k_B\,T_e}\nabla_{_S}\, \mu,\qquad v_n=-\Omega_0 (\nabla_{_S} \cdot \vec J)=\frac{D_s\nu\Omega_0}{k_B\,T_e}\nabla_{_S}^2\mu.
\label{eqn:svn2}
\end{equation}
In these expressions, $\vec J$ is the mass flux of atoms, $D_s$ is the surface diffusivity, $k_B\,T_e$ is the thermal energy, $\nu$ is the number of diffusing atoms per unit area, $\nabla_{_S}$ is the surface gradient. In addition to the surface diffusion which controlled the motion of the moving interface, we still need the boundary condition for the moving contact line. Following the idea for simulating solid-state dewetting in 2D \cite{Wang15,Jiang16,Jiang18b}, we assume that the normal velocity of the contact line $\Gamma$ is simply given by the energy gradient flow, which is determined by the time-dependent Ginzburg-Landau kinetic equations, i.e.,
\begin{equation}
v_{n_{_\Gamma}} = - \eta\frac{\delta W}{\delta \Gamma} = -\eta (\vec c_{_\Gamma}^\gamma\cdot\vec n_{_\Gamma}+\gamma_{_{\subFS}} - \gamma_{_{\subVS}}),
\label{eqn:gammavn}
\end{equation}
with $0<\eta<\infty$ denoting the contact line mobility, which can be thought of as a reciprocal of a constant
friction coefficient. For the physical explanation behind this approach, please refer to the recent paper~\cite{Wang15}.

We choose the characteristic length scale and characteristic surface energy scale as $h_0$ and $\gamma_{_0}$, respectively, the time scale as $\frac{h_0^4}{B\gamma_0}$ with $B=\frac{D_s\nu\Omega_0^2}{k_B\,T_e}$, and the contact line mobility is scaled by $\frac{B}{h_0^3}$. Let $\vec X(\cdot,t)=(x(\cdot,t),y(\cdot,t),z(\cdot,t))^T$ be a local parameterization of the moving film/vapor interface $S$, then we can obtain a dimensionless kinetic sharp-interface model for solid-state dewetting of thin film via the following Cahn-Hoffman $\boldsymbol{\xi}$-vector formulation as
\begin{eqnarray}
\label{eqn:weak31}
&&\partial_t\vec X=\Delta_{_S}\mu\;\vec n,\qquad t>0,\\
&&\mu=\nabla_{_S}\cdot\boldsymbol\xi,\qquad \boldsymbol\xi(\vec n)=\nabla\hat{\gamma}(\vec p)\Big|_{\vec p=\vec n},
\label{eqn:weak32}
\end{eqnarray}
where $t$ is the time, $\vec n$ is the unit outer normal vector of $S$, and $\boldsymbol{\xi}:=\boldsymbol{\xi}(\vec n)$ is the Cahn-Hoffman vector (scaled by $\gamma_0$). Here, for simplicity, we still use the same notations for all the dimensionless variables.

Let $\vec X_{_\Gamma}(\cdot,t)=(x_{_\Gamma}(\cdot,t),y_{_\Gamma}(\cdot,t),
z_{_\Gamma}(\cdot,t))^T$ represents a parametrization of the moving contact line $\Gamma(t)$. The initial condition is given as $S_0$ with boundary $\Gamma_0$ such that
\begin{equation}
S_0:=\vec X(\cdot,0)=(x_0,y_0,z_0),\quad  \Gamma_0:=\vec X(\cdot,0)\Big|_{\Gamma}.
\label{eqn:initialcondition}
\end{equation}
The above governing equations are subject to the following boundary conditions:

(i) contact line condition
\begin{equation}
z_{_\Gamma}(\cdot,t)=0,\quad t\geq 0;
\label{eqn:boundcon1}
\end{equation}

(ii) relaxed contact angle condition
\begin{equation}
\partial_t \vec X_{_\Gamma}=-\eta\Bigl(\vec c_{_\Gamma}^\gamma\cdot\vec n_{_\Gamma} - \sigma\Bigr)\vec n_{_\Gamma},\qquad t\geq0;
\label{eqn:boundcon2}
\end{equation}

(iii) zero-mass flux condition
\begin{equation}
\Bigl(\vec c_{_\Gamma}\cdot\nabla_{_S}\,\mu\Bigr)\Big|_\Gamma=0,\qquad t\geq0;
\label{eqn:boundcon3}
\end{equation}
where $\eta$ represents a (dimensionless) contact line mobility, $\vec c_{_\Gamma}^\gamma$ is the anisotropic co-normal vector which is defined as $\vec c_{_\Gamma}^\gamma:=(\boldsymbol{\xi}\cdot\vec n)\,\vec c_{_\Gamma}-(\boldsymbol{\xi}\cdot\vec c_{_\Gamma})\,\vec n$, $\vec c_{_\Gamma}$ represents the co-normal vector, and $\vec n_{_\Gamma}=(n_{_{\Gamma,1}},n_{_{\Gamma,2}},0)^T$ is the outer unit normal vector of $\Gamma$ on the flat substrate (cf.~Fig.~\ref{fig:dewetting3d}), and $\sigma = \frac{\gamma_{_{\subVS}}-\gamma_{_{\subFS}}}{\gamma_{_0}}$ is a (dimensionless) material constant.

\begin{remark}
For isotropic surface energy, i.e., $\gamma(\vec n)\equiv 1$, we obtain that $\boldsymbol{\xi}=\vec n$ and $\mu=\nabla_{_S}\cdot\boldsymbol\xi=\nabla_{_S}\cdot\vec n=\mathcal{H}$; for anisotropic surface energy,
by Definition~\ref{defi:calculus} and some calculations, we can obtain that the dimensionless chemical potential $\mu$ is the weighted mean curvature discussed in~\cite{Taylor92}.
\end{remark}

\begin{remark}
The contact line condition in~\eqref{eqn:boundcon1} ensures that the contact line must move along the substrate plane. Because the contact line $\Gamma$ lies on the substrate (i.e., $Oxy$ plane), the third component of $\vec n_{_\Gamma}$ is always zero, i.e., $n_{_{\Gamma,3}}=0$. As long as the initial condition satisfies $z_{_\Gamma}(\cdot,0)=0$, it can automatically satisfy the boundary condition (i) $z_{_\Gamma}(\cdot,t)=0,\forall\, t>0$ by using the boundary condition (ii).  The last boundary condition (iii) ensures that the total volume/mass of the thin film is conserved during the evolution, i.e., no-mass flux at the moving contact line.
\end{remark}

\begin{remark}
The above governing equation is well-posed when the surface energy is isotropic or weakly anisotropic. But when the surface energy is strongly anisotropic, some missing orientations will appear on equilibrium shapes~\cite{Sekerka05,Spencer04}; in this case, the governing equation becomes ill-posed, and it can be regularized by adding regularization terms such that the regularized sharp-interface model is well-posed~\cite{Jiang16,Bao17}. For the analytical criteria about the classification of surface energy anisotropy in 3D, interested readers could refer to~\cite{Sekerka05}.
\end{remark}

\subsection{Mass conservation and energy dissipation}

In the following, we will rigorously prove that the proposed sharp-interface model satisfies the mass conservation and the total free energy dissipation during the evolution.

\begin{prop}
Assume that $\vec X(\cdot,t)$ is the solution of the sharp-interface model, i.e., \eqref{eqn:weak31}-\eqref{eqn:weak32} with boundary conditions~\eqref{eqn:boundcon1}-\eqref{eqn:boundcon3}, and denote $S(t):=\vec X(\cdot,t)$ as the moving film/vapor interface. Then, the total volume (or mass) of the thin film, labeled as $|\Omega(t)|$, is conserved, i.e.,
\begin{equation}
|\Omega(t)|\equiv |\Omega(0)|,\qquad t\ge0.
\end{equation}
Furthermore, the (dimensionless) total interfacial free energy of the system is non-increasing during the evolution, i.e.,
\begin{equation}\label{eqn:totalenegw}
W(t)\le W(t_1)\le W(0)=\int_{S(0)}\gamma(\vec n)\;dS - \sigma A(\Gamma(0)),\qquad
 t\ge t_1\ge0.
\end{equation}
\end{prop}

\begin{proof}
By making use of the first variation~\eqref{eqn:domainvariation} and simply choosing the integrand
$\psi(\vec x)\equiv 1, \forall\,\vec x \in \Omega$, and using the governing equation~\eqref{eqn:weak31}, we can calculate the time derivative of the total volume as (noting that ${\vec V}_0=\partial_t\vec X$)
\begin{equation}
\frac{d}{dt}|\Omega(t)|=\int_{S(t)}\partial_t\vec X\cdot\vec n\;dS=\int_{S(t)}\Delta_{_S}\mu\;dS=0, \qquad t\ge0,
\end{equation}
where the last equality comes from the integration by parts and the zero-mass flux condition~\eqref{eqn:boundcon3}, and it
indicates that the total volume/mass is conserved.

To obtain the time derivative of the (dimensionless) total free energy, by making use of Theorem~\ref{thm:Wvariation}
and~\eqref{eqn:energyvariation}, but replacing the perturbation variable $\epsilon$ with the time variable $t$,
we can immediately obtain
\begin{eqnarray}
\frac{d}{dt}W(t)&=&\int_{S(t)}(\nabla_{_S}\cdot\boldsymbol{\xi})\,(\partial_t\vec X\cdot\vec n)\;dS+\int_{\Gamma(t)}(\vec c_{_\Gamma}^\gamma\cdot\vec n_{_\Gamma} - \sigma)\,(\partial_t\vec X_{\Gamma}\cdot\vec n_{_\Gamma})\;d\Gamma.\nonumber
\end{eqnarray}
By substituting the governing equations and the relaxed contact angle boundary condition, i.e.,
\begin{equation}
\mu=\nabla_{_S}\cdot\boldsymbol{\xi}, \quad \Delta_{_S}\mu = \partial_t\vec X\cdot\vec n,\quad \partial_t\vec X_{\Gamma}\cdot\vec n_{_\Gamma} = -\eta(\vec c_{_\Gamma}^\gamma\cdot\vec n_{_\Gamma} - \sigma),
\end{equation}
into the above equation and using the integration by parts and the zero-mass flux condition, we obtain
\begin{eqnarray}
\frac{d}{dt}W(t)&=&\int_{S(t)}\mu\,\Delta_S\mu\;dS-\eta\int_{\Gamma(t)}(\vec c_{_\Gamma}^\gamma\cdot\vec n_{_\Gamma} - \sigma)^2\;d\Gamma\nonumber\\
&=&-\int_{S(t)}|\nabla_{_S}\mu|^2\;dS-\eta\int_{\Gamma(t)}(\vec c_{_\Gamma}^\gamma\cdot\vec n_{_\Gamma} - \sigma)^2\;d\Gamma\le 0,\qquad t\ge0,
\end{eqnarray}
where the constant $\eta>0$. The last inequality immediately implies the energy dissipation.
\end{proof}

\begin{remark}
In the above proof, we need to calculate the time derivatives of the total volume and the total free energy. These two derivatives can be easily obtained by making use of the speed method and the first variation presented
in Section 2. In Section 2, we consider any type of smooth perturbations. In fact, a family of evolving interface
surfaces $\{S(t)\}_{t\geq0}$ can be also thought of as a type of perturbations, only by replacing the perturbation variable
$\epsilon$ with the time variable $t$. Therefore, the time derivatives can be directly obtained by using the first
variation of the total volume functional and the total free energy functional.
\end{remark}

\section{Numerical results}

In this section, we perform numerical simulations for solid-state dewetting in 3D to investigate the morphological evolution of thin films in various cases. We implement the parametric finite element method (PFEM)~\cite{Bao17,Bao18} for solving the proposed sharp-interface model in 3D. For the detailed introduction of numerical algorithms about PFEM in 3D, interested readers could refer to~\cite{Bao18}.

First, we focus on the case for isotropic surface energy, i.e., $\gamma(\vec n) \equiv 1$. We start with numerical examples for an initially, short cuboid island with $(4,4,1)$ representing its length, width and height, respectively (as shown in Fig.~\ref{fig:Ne1}(a)). The computational parameter is chosen as $\sigma=\cos({5\pi}/{6})$. In Fig.~\ref{fig:Ne1}, we show several snapshots of the morphology evolution for the short cuboid towards its equilibrium shape. As time evolves, the initial sharp corners and edges along the island become smoother and smoother (see Fig.~\ref{fig:Ne1}(b)), and finally the island film approaches a spherical shape as its equilibrium shape (see Fig.~\ref{fig:Ne1}(f)).
\begin{figure}[!htp]
\centering
\includegraphics[width=0.95\textwidth]{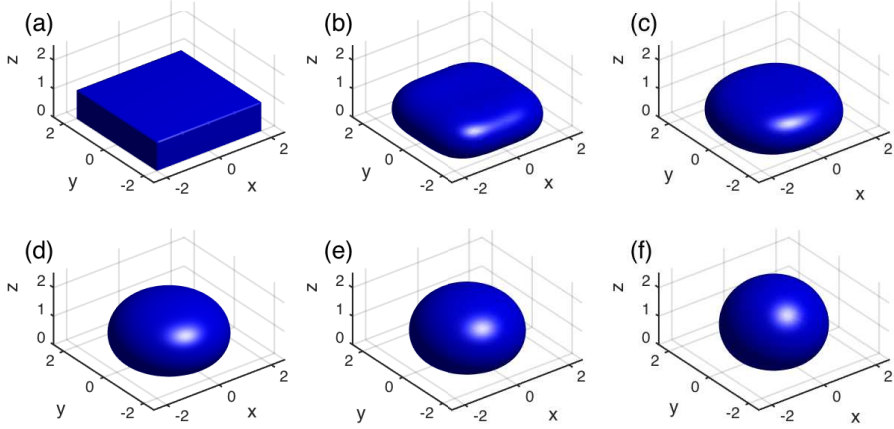}
\caption{Several snapshots during the evolution of an initially, cuboid island film with isotropic surface energy towards its equilibrium shape: (a) $t=0$; (b) $t=0.1$; (c) $t=0.2$; (d) $t=0.5$; (e) $t=0.7$; (f) $t=1.4$, where the initial shape of the thin film is chosen as a $(4,4,1)$ cuboid, and the material constant is chosen as $\sigma=\cos({5\pi}/{6})$.}
\label{fig:Ne1}
\end{figure}

\begin{figure}[!htp]
\centering
\includegraphics[width=0.95\textwidth]{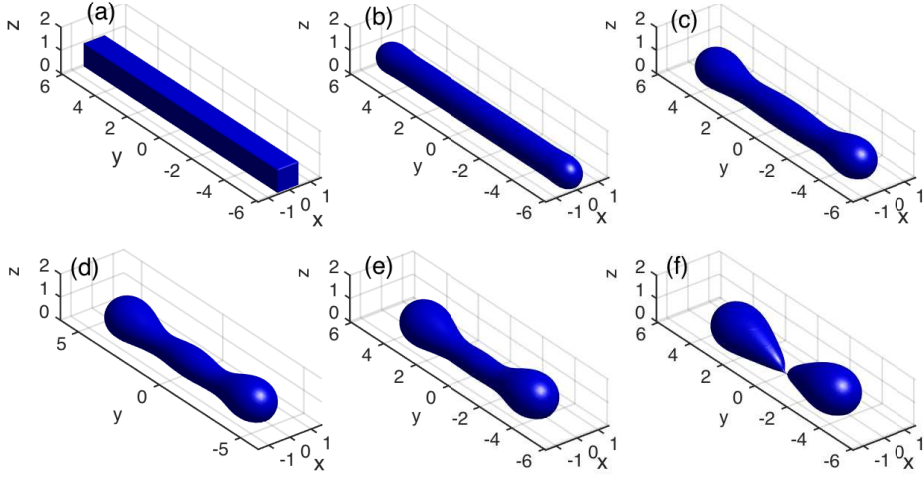}
\caption{Several snapshots during the evolution of an initial, cuboid island film with isotropic surface energy until its pinch-off time: (a) $t=0$; (b) $t=0.01$; (c) $t=0.30$; (d) $t=0.50$; (e) $t=0.80$; (f) $t=1.03$, where the initial shape is
chosen as a $(1,12,1)$ cuboid, and the material constant $\sigma=\cos (3\pi/4)$.}
\label{fig:1121islands}
\end{figure}

Short cuboid island films tend to form a single spherical shape as its equilibrium which minimizes its total free energy (i.e., the minimal surface area). However, the morphological evolution for long cuboid islands could be quite different. Due to the Plateau-Rayleigh instability~\cite{Kim15,Rayleigh78,Mccallum96}, long cuboid islands could pinch off and break up into a number of small isolated particles on the substrate before they approach a single spherical shape as its equilibrium. In order to investigate this phenomenon, we perform the simulation by choosing the material constant as $\sigma=\cos(3\pi/4)$, and the shape of initial island film as a long cuboid with $(1,12,1)$. For the isotropic case, as can be seen in Fig.~\ref{fig:1121islands}, the island quickly evolves into a cylinder-like shape during the evolution; then it accumulates more and more materials near the two edges, while the two necks appear and become thinner and thinner; finally, it pinches off at the neck and breaks up into two small isolated islands on the substrate. For cubic anisotropic surface energies, long cuboid islands also exhibit the similar pinch-off process as the isotropic surface energy case. We test the numerical example for an initially cuboid island with the same material constant and initial shape, as shown in Fig.~\ref{fig:A1121islands}. From the figure, we observe that three isolated small particles finally appear, while only two small particles are finally produced by the solid-state dewetting process in the isotropic case. This indicates that for this type of cubic anisotropic surface energy, the solid film tends to dewet more easily and quickly than in the isotropic case.

\begin{figure}[!htp]
\centering
\includegraphics[width=0.95\textwidth]{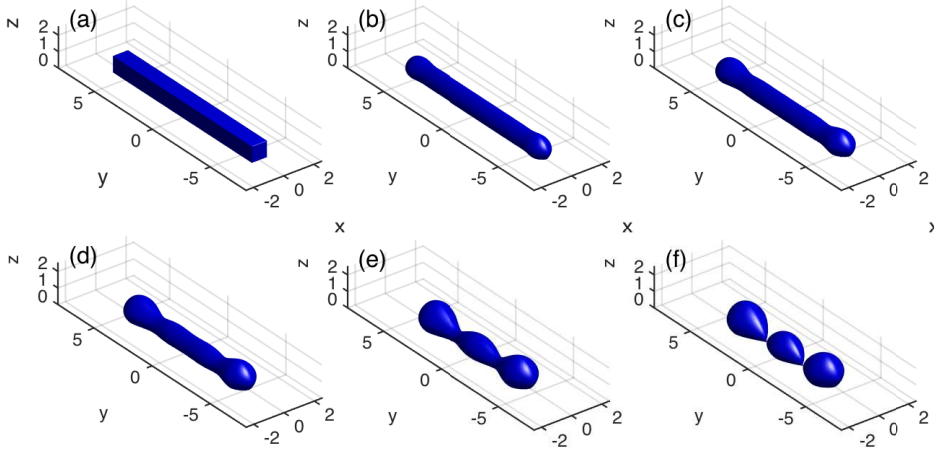}
\caption{Several snapshots during the evolution of an initial, cuboid island film with anisotropic surface energy until its pinch-off time: (a) $t=0$; (b) $t=0.020$; (c) $t=0.100$; (d) $t=0.240$; (e) $t=0.540$; (f) $t=0.695$, where the initial shape is chosen as a $(1,12,1)$ cuboid, the material constant $\sigma=\cos (3\pi/4)$, and the anisotropic surface energy is chosen as the cubic type, i.e., $\gamma(\vec n)=1+a(n_1^4+n_2^4+n_3^4)$ with $a=0.25$.}
\label{fig:A1121islands}
\end{figure}

\begin{figure}[!htp]
\centering
\includegraphics[width=0.95\textwidth]{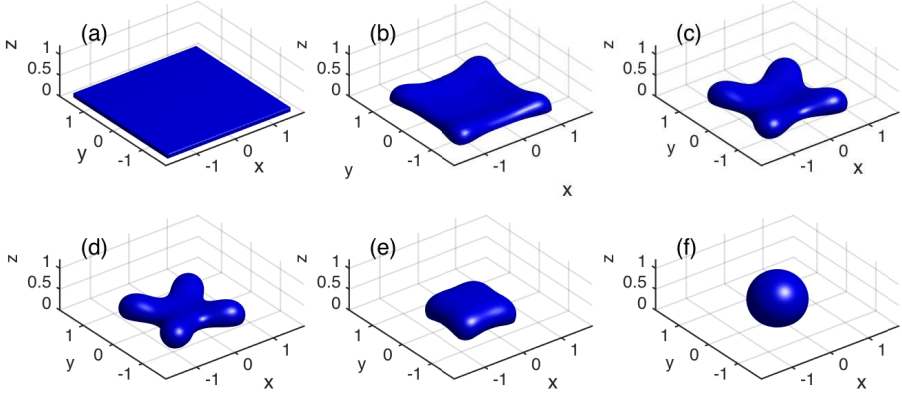}
\caption{Several snapshots during the evolution of an initial, small square island film with isotropic surface energy towards its equilibrium shape: (a) $t=0$; (b) $t=0.004$; (c) $t=0.008$; (d) $t=0.012$; (e) $t=0.020$; (f) $t=0.080$, where the initial shape is chosen as a $(3.2,3.2,0.1)$ cuboid, and the material constant $\sigma=\cos (5\pi/6)$.}
\label{fig:32321island}
\end{figure}

\begin{figure}[!htp]
\centering
\includegraphics[width=0.95\textwidth]{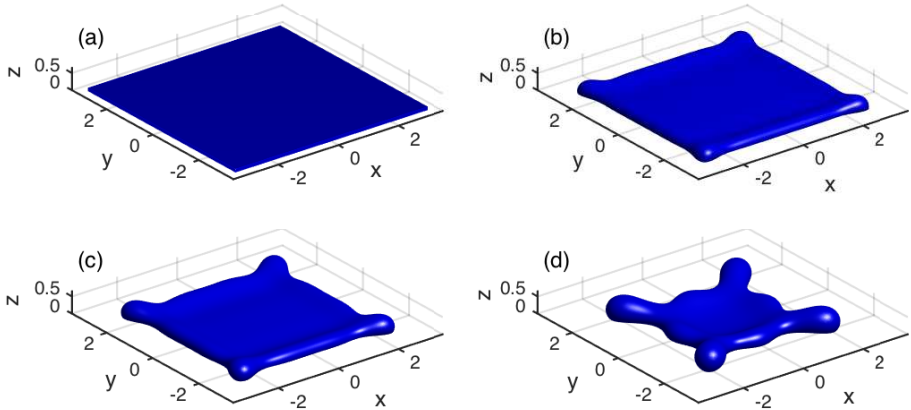}
\caption{Several snapshots during the evolution of an initial, large square island film with isotropic surface energy until its pinch-off time: (a) $t=0$; (b) $t=0.005$; (c) $t=0.010$; (d) $t=0.031$, where the initial shape is chosen as a $(6.4,6.4,0.1)$ cuboid, and the material constant $\sigma=\cos (5\pi/6)$.}
\label{fig:64641island}
\end{figure}

\begin{figure}[!htp]
\centering
\includegraphics[width=0.95\textwidth]{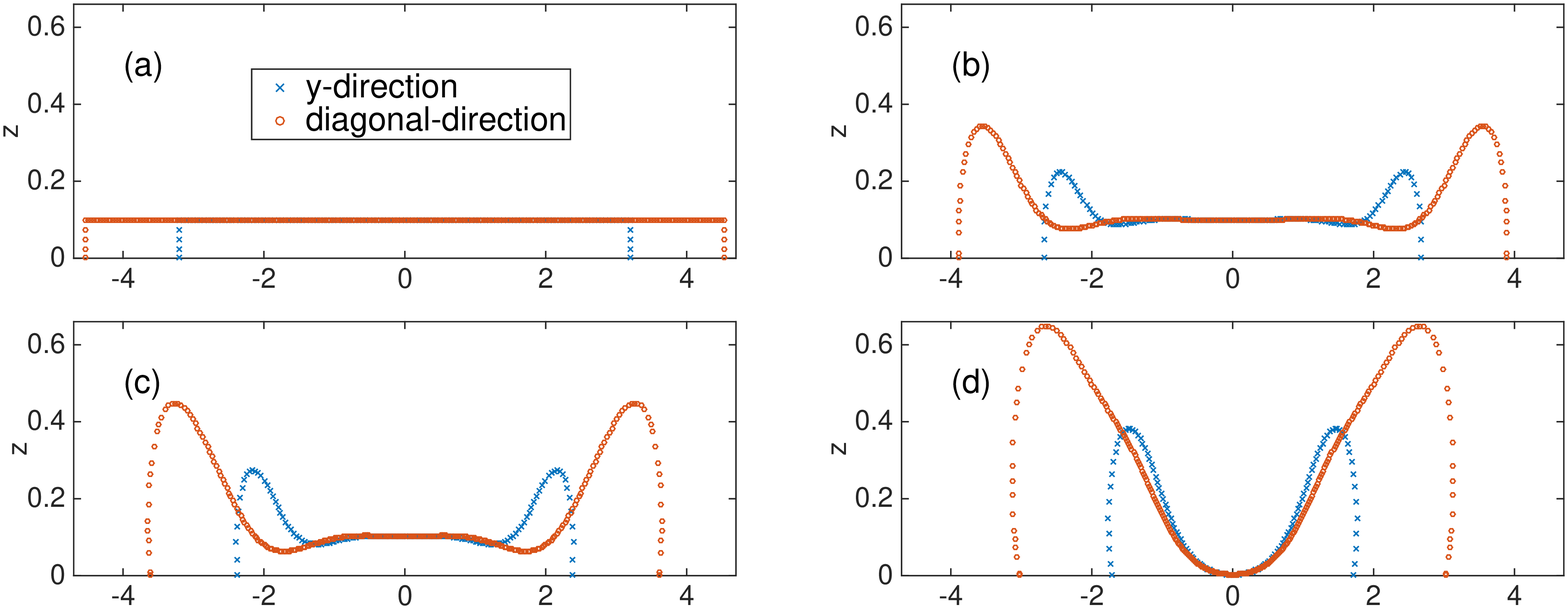}
\caption{The cross-section profile of the island film along its $y$-direction and diagonal direction for the example shown in Fig.~\ref{fig:64641island}: (a) $t=0$; (b) $t=0.005$; (c) $t=0.010$; (d) $t=0.031$.}
\label{fig:64641islandcc}
\end{figure}

Finally, we investigate the morphological evolution of square island films with size $(m,m,h)$ on a flat substrate. We start by simulating the evolution of an initial, small square island with size $(3.2,3.2,0.1)$, and the material constant is chosen as $\sigma=\cos({5\pi}/{6})$. As can be seen in Fig.~\ref{fig:32321island}, at the beginning, the square island retracts much more slowly at its four corners than at the middle points of the four edges. As time evolves, this process results in an almost cross shape (see Fig.~\ref{fig:32321island}(c)-(d)). This phenomenon, known as ``mass accumulation'' at the corner, has been previously observed in experiments~\cite{Thompson12,Ye10b,Ye11b} or numerical simulations by a phase-field approach~\cite{Jiang12,Naffouti17}. Subsequently, because the length of the square island is small, these retracting corners eventually catch up with the edges, then the contact line begins to move towards a circular shape in order to approach a spherical shape as its equilibrium (see Fig.~\ref{fig:32321island}(f)). During the evolution, we also observe that a valley appears at the center of the island, but finally it disappears. To observe the possible pinch-off phenomenon, we enlarge the square size and simulate the evolution of an initial, large square island with size $(6.4,6.4,0.1)$
(shown in Fig.~\ref{fig:64641island}). From the figure, we observe that the valley at the center becomes deeper and deeper, and
it eventually touches the substrate, and produces a hole in the center of the island. We stop the numerical simulation at the moment when there exists one new mesh point which touches the substrate. For a better illustration, in Fig.~\ref{fig:64641islandcc}, we also plot several snapshots about its cross-section profile
of the island film during the evolution.

\section{Conclusions}

We proposed a sharp-interface model for simulating solid-state dewetting of thin films in 3D, and this model can include the effect of the surface energy anisotropy.
Based on the Cahn-Hoffman $\boldsymbol{\xi}$-vector formulation and shape derivatives, we derived rigorously the first variation of the total free energy functional of the solid-state dewetting problem. From the first variation, necessary conditions for the equilibrium shape of solid-state dewetting were rigorously given in mathematics. Furthermore, a kinetic sharp-interface model was also proposed for simulating the solid-state dewetting of thin films in 3D. The governing equations described the interface evolution which is controlled by surface diffusion and contact line migration. Numerical simulations were performed by solving the proposed sharp-interface model, and numerical results reproduced the complex features in the solid thin film dewetting observed in experiments, such as edge retraction, hole formation, faceting, corner accumulation, pinch-off and Rayleigh instability.

In this paper, we assume that the surface diffusion is the only driving force for solid-state dewetting, and the effect of elasticity associated with a mismatch strain in the film layer during
deposition is negligible, which is often the case for Ni films on MgO substrates and Si films on
amorphous SiO$_2$ substrates (i.e., SOI). In the future, we can include the elastic effect into the model to study dewetting phenomena in semiconductor electronic and optoelectronic devices (such as In-GaAs/GaAs and SiGe/Si systems), especially for studying the Stranski-Krastanow (SK) and Volmer-Weber (VW) growth modes. Meanwhile, we can also study how the material microstructure affects the boundary conditions as discussed in~\cite{Tripathi18}.


\bibliographystyle{siamplain}
\bibliography{thebib}
\end{document}